%% file: LU_criterion_26v5.tex
\documentclass[ 
 superscriptaddress,
 longbibliography,
 onecolumn,
 amsmath,amssymb,
 aps,
 pra,
]{revtex4-2}

\input{notation}

\begin{document}


\title{Characterizing local unitary equivalence of multipartite states by local unitary Bargmann invariants} 

\author{Yi Shen}
\email[]{yishen@jiangnan.edu.cn}
\affiliation{School of Mathematics and Data Science, Jiangnan University, Wuxi Jiangsu 214122, China}

\author{Lin Chen}
\email[Corresponding author: ]{linchen@buaa.edu.cn}
\affiliation{School of Mathematical Sciences, Beihang University, Beijing 100191, China}

\author{Shao-Ming Fei}
\email[Corresponding author: ]{feishm@cnu.edu.cn}
\affiliation{School of Mathematical Sciences, Capital Normal University, Beijing, 100048, China}

\date{\today} 

\begin{abstract}
Local unitary (LU) equivalent states are crucial for understanding the structure of quantum entanglement. We investigate the LU equivalence based on the LU Bargmann invariants (LUBIs). We consider for which states the LUBIs are sufficient to ensure the LU equivalence. We focus on a class of multipartite states with complete local commutativity, called CLC states. We first show that for multipartite CLC states, all the LUBIs are sufficient to ensure the local permutation equivalence, if their single-party marginals are diagonal and non-degenerate. By virtue of the fact that the CLC property remains unchanged under local operations, we verify that the LUBIs are also sufficient to ensure the LU equivalence of multipartite CLC states whose single-party marginals are non-degenerate. Furthermore, when the single-party marginals are degenerate, we propose concrete examples of two-qudit and three-qubit systems, showing that there are states which share all the LUBIs but are not LU equivalent.
\end{abstract}


\maketitle



\section{Introduction}

Quantum entanglement is the quintessential non-classical resource in quantum information science and underpins a wide range of protocols including quantum teleportation \cite{telepor-naturereview2023}, superdense coding \cite{densecoding2022} and quantum cryptography \cite{qcrypto2020}. A central task in the study of quantum entanglement is to classify quantum states according to their equivalence under local unitary (LU) transformations, since states within the same equivalence class are fully interchangeable in any local-operation-based quantum information processing tasks \cite{3qubitlu2000,3qubitinequiv2001}. Two LU-equivalent states share the identical quantum correlations such as entanglement, the local distinguishability and the performance in quantum protocols \cite{luequiv2010,pslu-2010}.

For bipartite pure states, the LU equivalence problem is fully solved by virtue of the Schmidt decomposition. The problem becomes significantly complex for multipartite systems and mixed states. For multipartite pure states, Kraus \cite{luequiv2010} derived a necessary and sufficient criterion for LU equivalence based on the spectra of reduced density matrices and the associated correlation tensors, enabling the complete LU classification up to five-qubit pure states. Subsequent work has developed alternative invariant-based approaches, including polynomial invariants constructed from the generalized Schmidt decompositions \cite{mpsinequivlu2012} and singular values of spin-flipping matrices \cite{luspinflip2018}, which provide computationally tractable tools for specific classes of pure states. 

The LU equivalence of mixed states is substantially challenging due to the convex structure of the state space and spectral degeneracies. For two-qubit mixed states, Makhlin \cite{Makhlin2002} established a complete set of 18 polynomial invariants that fully determine the LU equivalence, later refined by Sun et al. who demonstrated that only 14 invariants are required for arbitrary two-qubit states \cite{3qbitinv2017}. For higher-dimensional and multipartite mixed states, progress has been made towards using techniques including matrix algebra \cite{bilu2012}, matrix realignment \cite{msinv2013}, partial transposition \cite{PTLU2013}, polynomial invariants in terms of the generalized Bloch representation \cite{mqbitinv2015}, as well as hyperdeterminants \cite{luhypermatrix2025}, with recent works extending these methods to handle degenerate spectra \cite{LUdegen2026}. Notably, the LU equivalence problem also connects to operational questions in quantum information. For instance, the LU-equivalent sets of generalized Bell states exhibit identical local distinguishability under local operation and classical communication (LOCC) \cite{LUGBS2025}. The classification of LU equivalence classes of entanglement witnesses provides a dual perspective on state equivalence \cite{yiprr2025}.

Despite these advances, it still lacks a complete set of LU invariants for generic multipartite states, and a practical implementation of classification algorithms for higher-dimensional systems. The Bargmann invariants, the fundamental unitary-invariants, have a wide range of applications including quantum fingerprinting \cite{qfinprint2001}, Kirkwood-Dirac quasi-probability distributions \cite{KDBI2024} and quantum imaginarity witnesses and measures \cite{BIIWprl2024,BIIWpra2025,BIIWpra2026,SetIguo2026}. Also referred to as multivariate traces, these quantities admit estimation through constant-depth circuits, thereby ensuring both compatibility with near-term hardware and experimental feasibility \cite{multivariatetr2024}. For the above reasons, the characterization of Bargmann invariants has attracted a great interest recently \cite{BIchara202504,BIchara202510,BIcharaXU202601,BIcharaXu202602,BIZL2026review}.
As a counterpart, the LU Bargmann invariants (LUBIs), a kind of LU invariants, were recently introduced to detect entanglement and classify multipartite states \cite{BIZL202501}. 

In this paper we investigate whether the set of LUBIs are sufficient to ensure the LU equivalence. We show that the LUBIs are sufficient to certify the LU equivalence for special classes of states, but not for general states \cite{LUBI202607}.

We introduce the concept of local commutativity to simplify the criterion in terms of the LUBIs and define the locally commutative (LC) states with respect to some subsystem by Def. \ref{def:commute}, and the completely locally commutative (CLC) state is defined as the one which is locally commutative with respect to each single-party subsystem $A_j$. Then we characterize the LU orbits of LC states by the LUBIs. In Lemma \ref{le:com-inv} we show that the locally commutative property remains unchanged under LU equivalence. It allows us to apply arbitrary LU operations on the states in order to determine whether they are LU equivalent. As a corollary of Lemma \ref{le:pcommute} we further show that all the LUBIs ensure the local commutativity with respect to the same subsystem.
Based on the characterization of LC states we consider testing LU equivalence between CLC states by the LUBIs. For bipartite CLC states $\rho_{AB}$ and $\sigma_{AB}$, the matched LUBIs can be simplified as the standard form: 
\beq
\label{eq:bi-sta}
\tr[(\rho_{AB})^p(\rho_A\ox\I_B)^q(\I_A\ox\rho_B)^r]=\tr[(\sigma_{AB})^p(\sigma_A\ox\I_B)^q(\I_A\ox\sigma_B)^r], ~~\forall p,q,r \geq 0.
\eeq
With this standard form we show in Proposition \ref{thm:distinct} that for a bipartite CLC state $\rho_{AB}$ with two non-degenerate diagonal marginals, (i) the state that is LU equivalent to $\rho_{AB}$ is indeed LP equivalent to $\rho_{AB}$, and (ii) the state $\sigma_{AB}$ is LU equivalent to $\rho_{AB}$ if and only if they share all the matched LUBIs. We also propose a pair of diagonal states in Example \ref{ex:adv} to demonstrate the discriminatory power of LUBIs relative to the common standard necessary invariants. With the fact from Lemma \ref{le:com-inv} that the LC property remains unchanged after applying LU operations, we generalize the above result to the bipartite states whose marginals may not be diagonal. We show in Theorem \ref{thm:comm+sim} that for a CLC state $\rho_{AB}$ with two non-degenerate marginals, another state $\sigma_{AB}$ is LU equivalent to $\rho_{AB}$ if and only if they share all the matched LUBIs. In Ref. \cite{yiprr2025} the connection between the LU equivalence relations on the sets of states and the set of EWs was studied. We are inspired to extend the above results to determine the LU equivalence between bipartite EWs. In Proposition \ref{prop:biew} we similarly show that for bipartite CLC EWs, they are LU equivalent if and only if they share all the matched LUBIs generated from the corresponding EWs. Then we generalize the above two results to the case of multipartite CLC states by Proposition \ref{le:tri-lp} and Theorem \ref{thm:tri-lu}. Specifically, for two multipartite CLC states with non-degenerate single-party marginals, all matched LUBIs are sufficient to ensure they are LU equivalent. It is worth mentioning that at most $D$ (equal to the global dimension) LUBIs are needed to verify the LU equivalence between such states, which means this criterion is operational and resource-saving for such states.

When testing LU equivalence of CLC states, we realize that the precondition of non-generate single-party marginals are crucial for the sufficiency of the proposed criterion. Specifically we indicate by examples that the criterion in terms of all matched LUBIs is not sufficient to ensure LU equivalence if the single-party marginals could be degenerate. It is known that the criterion in terms of matched LUBIs is necessary and sufficient for two-qubit states \cite{bilu2012}. We thus construct two-qutrit states in Example \ref{ex:c-ex} which share all the matched LUBIs but are not LU equivalent. We then extend such an example to two-qudit states with local dimension greater than $2$, by Example \ref{ex:cex-2}. Finally we reveal by Example \ref{ex:cex-tri} that such a pair of states also exists in multipartite systems, e.g., there exist three-qubit states $\rho,\sigma$ which share all the matched LUBIs but are not LU equivalent.

The remainder of this paper is organized as follows: In Sec. \ref{sec:pre}, we clarify the notations and definitions, and present a necessary and sufficient condition for determining the LU equivalence. We then test the LU equivalence of multipartite CLC states which share all the matched LUBIs in Sec. \ref{sec:clc}. In Sec. \ref{subsec:clc-c} we characterize the locally commutative states with all matched LUBIs under LU equivalence. In Sec. \ref{subsec:bi-lu} and Sec. \ref{subsec:multi-lu} we respectively show that the criterion in terms of the matched LUBIs is effective to distinguish LU orbits for bipartite and multipartite CLC states. In Sec. \ref{sec:cex} we propose concrete examples to show that the sufficiency of this criterion fails when the states have degenerate single-party marginals. Finally, the concluding remarks are given in Sec. \ref{sec:con}.

\section{Preliminaries}
\label{sec:pre}

Let $[n]:=\{1,2,\cdots,n\}$. We denote by $[n]^l$ the Cartesian product of $l$ times of $[n]$. For an $n$-partite system of $A_1\cdots A_n$, we refer to $A_\cS$ as a subsystem of $A_{j_1}\cdots A_{j_k}$ with $\cS=\{j_1,\cdots,j_k\}$, and correspondingly $A_{\cS^c}$ represents the subsystem with $\cS^c=[n]\backslash\cS$. Let $\rho$ be an $n$-partite state of the system $A_1\cdots A_n$. Denote by $\rho_{A_\cS}$ the marginal of $\rho$ with respect to the subsystem $A_\cS$. 
An $n$-partite operator $P$ is called local if $P=\bigotimes_{j=1}^n P_j$, where each $P_j$ is a permutation operator on the $j$th subsystem. Two $n$-partite operators $X$ and $Y$ are called equivalent by local permutation, also known as LP equivalent, if there exists a local permutation $P$ such that $Y=PXP^\dg$. 

We focus on characterizing LU equivalence of states by the local unitary Bargmann invariants which are the multivariate traces generated from the states and their single-party marginals. 
Given a tuple of Hermitian operators supported on the same Hilbert space, $(X_1,\cdots,X_n)$, we refer to the Bargmann invariants of this tuple associated with a string $(j_1,\cdots,j_l)$ as the $(j_1,\cdots,j_l)$-BI, $\tr(X_{j_1}X_{j_2}\cdots X_{j_l})$, where each $j_i$ is from $[n]$.
The characterization of LU equivalence is indeed to identity the LU invariants that distinguish  different LU orbits. 


\begin{definition}
\label{def:match}
Two $n$-partite states $\rho_{A_1\cdots A_n}$ and $\sigma_{A_1\cdots A_n}$ are matched, if for all strings $(j_1,\cdots,j_l)\in[2^n-1]^l$ the associated LUBIs from $\rho_{AB}$ and $\sigma_{AB}$ match, i.e.,
 \beq
\label{eq:cj-2m}
\tr\left[\alpha_{j_1}\alpha_{j_2}\cdots\alpha_{j_l}\right]=\tr\left[\beta_{j_1}\beta_{j_2}\cdots\beta_{j_l}\right],
\eeq 
where 
\beq
\label{eq:cj-2m.1}
\bal
\{\a_1,\cdots,\a_{2^n-1}\}&=\{\rho_{A_{\cS}}\ox\I_{A_{\cS^c}}~|~\emptyset\neq\cS\subseteq[n]\}, \\
\{\b_1,\cdots,\b_{2^n-1}\}&=\{\sigma_{A_{\cS}}\ox\I_{A_{\cS^c}}~|~\emptyset\neq\cS\subseteq[n]\}.
\eal
\eeq
\end{definition}

In particular, for two bipartite states $\rho_{AB}$ and $\sigma_{AB}$, they are matched if for all strings $(j_1,\cdots,j_l)\in[3]^l$, the associated LUBIs from $\rho_{AB}$ and $\sigma_{AB}$ match,
\beq
\label{eq:cj-1}
\tr\left[\alpha_{j_1}\alpha_{j_2}\cdots\alpha_{j_l}\right]=\tr\left[\beta_{j_1}\beta_{j_2}\cdots\beta_{j_l}\right],
\eeq 
where $\a_1=\rho_{AB},\a_2=\rho_A\ox\I_B,\a_3=\I_A\ox\rho_B$, and $\b_1=\sigma_{AB},\b_2=\sigma_A\ox\I_B,\b_3=\I_A\ox\sigma_B$.


\begin{definition}
\label{def:commute}
For an $n$-partite state $\rho_{A_1\cdots A_n}$, if it commutes with $\rho_{A_\cS}\ox\I_{A_{\cS^c}}$, then $\rho_{A_1\cdots A_n}$ is called locally commutative (LC) with respect to the subsystem $A_\cS$. Furthermore, if $\rho_{A_1\cdots A_n}$ commutes with $\rho_{A_\cS}\ox\I_{A_{\cS^c}}$ for any set $\cS$ with $\abs{\cS}=1$, then it is called completely locally commutative (CLC). 
\end{definition}


We first propose a necessary and sufficient condition for two LU equivalent multipartite Hermitian operators, which applies not only to multipartite states, but also to non-positive semidefinite observables such as the entanglement witnesses (EWs).

\begin{lemma}
\label{le:lu-equiv-g}
Let $H=\sum_{j=1}^m\lambda_j P_j$ and $K=\sum_{j=1}^n \mu_jQ_j$ be the spectral decompositions of two $n$-partite Hermitian operators supported on $\bigotimes_{j=1}^n\cH_{A_j}$, where $\lambda_1>\lambda_2>\cdots>\lambda_m$ ($\mu_1>\mu_2>\cdots>\mu_n$) are the distinct eigenvalues of $H$ ($K$), $P_j$'s are pairwisely orthogonal projectors, $P_jP_k=\delta_{jk}P_j$ (similarly for $Q_j$'s). Then $H$ and $K$ are LU equivalent if and only if the following conditions both hold:
(i) $m=n$ and $\lambda_j=\mu_j$ for each $1\leq j\leq m$;
(ii) there exists a common local unitary $\bigotimes_{j=1}^n U_j$ such that $P_j=(\bigotimes_{j=1}^n U_j)Q_j(\bigotimes_{j=1}^n U_j)^\dg$ for each $1\leq j\leq m$.
\end{lemma}

\begin{proof}
The ``If'' part can be verified directly. We show the ``Only if'' part in detail. Suppose that $H$ and $K$ are LU equivalent. Since the eigenvalues are unitary invariants, condition (i) must hold. Let $m$ be the number of distinct eigenvalues, and $\lambda_j$'s the common eigenvalues of both $H$ and $K$. By definition there is a local unitary $\bigotimes_{j=1}^n U_j$ such that $H=(\bigotimes_{j=1}^n U_j)K(\bigotimes_{j=1}^n U_j)^\dg$. Based on the spectral decomposition of $K$, we obtain that
\beq
\label{eq:lu-equiv-g1}
H=\left(\bigotimes_{j=1}^n U_j\right)K\left(\bigotimes_{j=1}^n U_j\right)^\dg=\sum_{j=1}^m\lambda_j\left(\bigotimes_{j=1}^n U_j\right)Q_j\left(\bigotimes_{j=1}^n U_j\right)^\dg.
\eeq
One can verify that the decomposition given by Eq. \eqref{eq:lu-equiv-g1} is a spectral decomposition of $H$, as the operators with the form $(\bigotimes_{j=1}^n U_j)Q_j(\bigotimes_{j=1}^n U_j)^\dg$ remain to be pairwisely orthogonal projectors. For any Hermitian operator, the eigenspace corresponding to each distinct eigenvalue is unique, and thus the projector onto each eigenspace is unique. Therefore, for each eigenvalue $\lambda_j$ of $H$, the corresponding projector is unique,
\beq
\label{eq:commonlu}
P_j=\left(\bigotimes_{j=1}^n U_j\right)Q_j\left(\bigotimes_{j=1}^n U_j\right)^\dg.
\eeq
This implies condition (ii).
\end{proof}

\section{Identifying LU equivalence of CLC states with matched LUBIs}
\label{sec:clc}

We now consider whether all matched LUBIs guarantee the LU equivalence by focusing on the CLC states invloving the complete local commutativity. In Sec. \ref{subsec:clc-c} we characterize the locally commutative states with all matched LUBIs under LU equivalence. In Sec. \ref{subsec:bi-lu} and Sec. \ref{subsec:multi-lu} we respectively show that it is effective to distinguish LU orbits for bipartite CLC states and multipartite CLC states, by virtue of the matched LUBIs.  

\subsection{Characterization of LC states with all matched LUBIs}
\label{subsec:clc-c}

We first show that the locally commutative property proposed in Def. \ref{def:commute} remains unchanged under LU equivalence.

\begin{lemma}
    \label{le:com-inv}
Let $\rho_{A_1\cdots A_n}$ be an $n$-partite state. For any subsystem $A_j$, $\rho_{A_1\cdots A_n}$ commutes with $\rho_{A_j}\ox\I_{A_{j^c}}$, if and only if $(\bigox_{j=1}^n U_j)\rho_{A_1\cdots A_n}(\bigox_{j=1}^n U_j)^\dg$ commutes with $(U_j\rho_A U_j^\dg)\ox\I_{A_{j^c}}$ for an arbitrary local unitary $\bigox_{j=1}^n U_j$, where $A_{j^c}$ denotes the $(n-1)$-party subsystem without $A_j$.
\end{lemma}

\begin{proof}
As for a bipartion, we only need consider the bipartite state $\rho_{AB}$ and the corresponding local commutativity with respect to subsystem $A$. For an arbitrary local unitary $U\ox V$, by calculation we have
\beq
\label{eq:com-inv-1}
\bal
&\left((U\ox V)\rho_{AB}(U\ox V)^\dg\right)\left((U\rho_A U^\dg)\ox \I_B\right) \\
=&(U\ox V)\left(\rho_{AB}(\rho_A\ox\I_B)\right)(U\ox V)^\dg
\eal
\eeq
and
\beq
\label{eq:com-inv-2}
\bal
&\left((U\rho_A U^\dg)\ox \I_B\right)\left((U\ox V)\rho_{AB}(U\ox V)^\dg\right) \\
=&(U\ox V)\left((\rho_A\ox\I_B)\rho_{AB}\right)(U\ox V)^\dg.
\eal
\eeq
By the two equations above we conclude that $(U\ox V)\rho_{AB}(U\ox V)^\dg$ commutes with $(U\rho_A U^\dg)\ox\I_B$ if and only if $\rho_{AB}$ commutes with $\rho_A\ox\I_B$. This completes the proof.
\end{proof}

As a corollary of the above lemma, two $n$-partite states $\rho$ and $\sigma$ of system $A_1\cdots A_n$ are not LU equivalent if $\rho$ is an LC state with respect to some subsystem $A_j$ while $\sigma$ is not locally commutative with respect to the subsystem $A_j$. It actually provides a necessary condition to distinguish the LU orbits of LC states. In light of this fact, it is natural to ask if all the matched LUBIs ensure the local commutativity with respect to the same subsystem. If the answer is negative, the matched LUBIs cannot ensure the corresponding two states in the same LU orbit. Specifically, for a state $\rho_{A_1\cdots A_n}$ that is locally commutative with respect to some subsystem $A_j$, and another state $\sigma_{A_1\cdots A_n}$ matching with $\rho_{A_1\cdots A_n}$, we consider if $\sigma_{A_1\cdots A_n}$ is also locally commutative with respect to the same subsystem $A_j$. The following lemma implies that the answer is positive, and supports that the matched LUBIs can be used to test LU equivalence.

\begin{lemma}
    \label{le:pcommute}
    For two Hermitian matrices $M$ and $N$ supported on the same Hilbert space, they are commutative, $[M,N]=0$, if and only if they satisfy that
    \beq
    \label{eq:pcomm-1}
    \tr\left((MN)^2\right)=\tr\left(M^2 N^2\right).
    \eeq
\end{lemma}

\begin{proof}
The ``Only if" part can be verified directly. We next show the ``If" part.
   Let $C=[M,N]=MN-NM$. It follows that $C^\dg=-C$. Thus, we obtain that
   \beq
       \label{eq:pcomm-2}
       \tr(C^\dg C)=\tr(-CC)=-\tr(C^2).
   \eeq
   By calculation we obtain that
   \beq
\label{eq:pcomm-3}
\bal
&\tr(C^2)\\
=&\tr\left((MN)^2-(MNNM)-(NMMN)+(NM)^2\right) \\
=&\tr\left((MN)^2\right)-2\tr(M^2N^2)+\tr\left((NM)^2\right) \\
=&\tr\left((NM)^2\right)-\tr\left((MN)^2\right)=0,
\eal
\eeq
where the second equality follows from the cyclic law of trace, the third equality follows by Eq. \eqref{eq:pcomm-1}, and the last equality also follows from the cyclic law of trace. Then we conclude from Eq. \eqref{eq:pcomm-2} that $\tr(C^\dg C)=0$. Due to the semi-definite positivity of $C^\dg C$, we obtain that $C=0$, namely, $MN=NM$. This completes the proof.
\end{proof}

We take the bipartite case as an example to answer the question right above Lemma \ref{le:pcommute}.
We claim that a part of all the matched LUBIs formulated by Eq. \eqref{eq:cj-1} can ensure that $\beta_j$ and $\beta_k$ are commutative, if $\alpha_j$ and $\alpha_k$ are commutative. The reason is as follows. Assume that $\alpha_j$ commutes with $\alpha_k$. Then we obtain that 
\beq
\label{eq:pcommu-4}
\tr\left((\beta_j\beta_k)^2\right) =\tr\left((\alpha_j\alpha_k)^2\right) =\tr\left(\alpha_j^2\alpha_k^2\right) =\tr\left(\beta_j^2\beta_k^2\right).
\eeq
It follows from Lemma \ref{le:pcommute} that Eq. \eqref{eq:pcommu-4} implies $\beta_j\beta_k=\beta_k\beta_j$. It means that if $\rho_{AB}$ commutes with $\rho_A\ox \I_B$, the matched LUBIs ensure that $\sigma_{AB}$ also commutes with $\sigma_A\ox \I_B$.
This fact can be directly extended to the multipartite case. In other words, for any two multipartite states that are matched, if one state is locally commutative with respect to some given subsystem, then the other must be locally commutative with respect to the same subsystem. Thus, two states that are matched must satisfy the necessary condition of LU equivalence implied from Lemma \ref{le:com-inv}. Specially, any two multipartite states that are matched either are both CLC states or both are not.

\subsection{Identifying bipartite LU equivalence of CLC states}
\label{subsec:bi-lu}

In this subsection we characterize the LU orbits of bipartite CLC states with the matched LUBIs. For this purpose, we first study the marginals of $\rho_{AB}$ and $\sigma_{AB}$ which generate those $\a_2=\rho_A\ox\I_B$, $\a_3=\I_A\ox\rho_B$, and $\b_2=\sigma_A\ox\I_B$, $\b_3=\I_A\ox\sigma_B$. It suffices to assume that $\rho_A$ ($\rho_B$) shares the same eigenvalues as $\sigma_A$ ($\sigma_B$). Thus, by virtue of Lemma \ref{le:com-inv}, we may apply some LU operators and local permutations on the two bipartite states $\rho_{AB}$ and $\sigma_{AB}$, such that the marginals $\rho_A=\sigma_A$ and $\rho_B=\sigma_B$ are all diagonal. We focus on the case where the diagonal marginals $\rho_A=\sigma_A$ and $\rho_B=\sigma_B$ are all non-degenerate, i.e., the diagonal entries are distinct.

\begin{lemma}
\label{le:diagonal}
(i) For a bipartite Hermitian operator $M_{AB}$ supported on $\bbC^m\ox\bbC^n$, if it commutes with both $X_A\ox\I_B$ and $\I_A\ox Y_B$, where $X_A=\tr_B M_{AB}$ and $Y_B=\tr_A M_{AB}$ are some non-degenerate diagonal matrices acting on systems $A$ and $B$ respectively, then $M_{AB}$ has to be diagonal.

(ii) Suppose that $\rho_{AB}$ and $\sigma_{AB}$ are two bipartite CLC states supported on $\bbC^m\ox\bbC^n$, whose marginals are all non-degenerate diagonal. Then $\rho_{AB}$ and $\sigma_{AB}$ are LU equivalent if and only if they are the same up to a local permutation, i.e., they are LP equivalent.
\end{lemma}

\begin{proof}
(i) We may write $M_{AB}$ in a block matrix as
\beq
\label{eq:diag-m-1}
M_{AB}=
\bma
A_{11} & A_{12} & \cdots & A_{1m} \\
A_{12}^\dg & A_{22} & \cdots & A_{2m} \\
\vdots & \vdots & \cdots & \vdots \\
A_{1m}^\dg & A_{2m}^\dg & \cdots & A_{mm}
\ema,
\eeq
where each $A_{ij}$ is an $n\times n$ matrix. Assume $X_A=\diag(a_1,a_2,\cdots,a_m)$ and $Y_B=\diag(b_1,b_2,\cdots,b_n)$, with distinct $a_1,a_2,\cdots,a_m$ and $b_1,b_2,\cdots,b_m$. Due to the condition $[X_A\ox\I_B,M_{AB}]=0$ and the distinct elements $a_1,a_2,\cdots,a_m$, we conclude that the non-diagonal blocks are all zero matrices by comparing $(X_A\ox\I_B)M_{AB}$ with $M_{AB}(X_A\ox\I_B)$. That is, $A_{ij}=O$ for $i\neq j$ in Eq. \eqref{eq:diag-m-1}. Similarly, due to the condition $[\I_A\ox Y_B,M_{AB}]=0$ and the distinct elements $b_1,b_2,\cdots,b_n$, for each $A_{jj}$, the non-diagonal entries of $A_{jj}$ are all zero by comparing $(\I_A\ox Y_B) M_{AB}$ with $M_{AB}(\I_A\ox Y_B)$. Thus, $M_{AB}$ in \eqref{eq:diag-m-1} is a diagonal matrix supported on $\bbC^m\ox\bbC^n$. 

(ii) By assertion (i) we obtain that $\rho_{AB}$ and $\sigma_{AB}$ are both diagonal matrices when the corresponding $\alpha_j,\alpha_k$ are commutative, and $\beta_j,\beta_k$ are commutative, for each pair $(j,k)$. The ``if'' part is obvious. We next show the ``only if'' part. Assume that $\rho_{AB}$ is LU equivalent to $\sigma_{AB}$, namely, there exists an LU operator $U\ox V$ such that $(U\ox V)\rho_{AB}(U\ox V)^\dg=\sigma_{AB}$. Since the diagonal $\rho_A$ and $\sigma_A$ ($\rho_B$ and $\sigma_B$) share the identical eigenvalues, there exist two permutations $P$ and $Q$ such that $\sigma_A=P\rho_A P^\dg$ and $\sigma_B=Q\rho_B Q^\dg$. By the LU equivalence of $\rho_{AB}$ and $\sigma_{AB}$ we obtain that $U\rho_A U^\dg=\sigma_A=P\rho_A P^\dg$ and $V\rho_B V^\dg =\sigma_B=Q\rho_B Q^\dg$. Since $\rho_A$ is a diagonal matrix with distinct diagonal entries, one can verify that $U\rho_A U^\dg=P\rho_A P^\dg$ holds only if $P^\dg U$ is a diagonal unitary matrix. Similarly, we conclude that $V\rho_B V^\dg =Q\rho_B Q^\dg$ holds only if $Q^\dg V$ is a diagonal unitary matrix. That is, $U=P\Lambda_A$ and $V=Q\Lambda_B$ for some diagonal unitaries $\Lambda_A$ and $\Lambda_B$. It follows by the LU equivalence that  
\beq
\label{eq:inv-2}
\bal
\sigma_{AB}=&\left((P\ox Q)(\Lambda_A\ox\Lambda_B)\right)\rho_{AB}\left((\Lambda_A\ox\Lambda_B)^\dg(P\ox Q)^\dg\right) \\
=&(P\ox Q)\rho_{AB}(\Lambda_A\ox\Lambda_B)(\Lambda_A\ox\Lambda_B)^\dg (P\ox Q)^\dg \\
=&(P\ox Q)\rho_{AB}(P\ox Q)^\dg,
\eal
\eeq
namely, $\rho_{AB}$ and $\sigma_{AB}$ are the same up to a local permutation $P\ox Q$.
This completes the proof.
\end{proof}

Lemma \ref{le:diagonal} (ii) shows that the LU orbit of a bipartite CLC state with non-degenerate diagonal marginals is precisely the same as the LP orbit, which gives a more efficient method to test the LU equivalence of such states due to the relatively simple mathematical structure of permutation matrices \cite{LPLU25}.

\begin{proposition}
\label{thm:distinct}
Suppose that $\rho_{AB}$ is a CLC state with two non-degenerate diagonal marginals. For some state $\sigma_{AB}$ with two diagonal marginals, $\rho_{AB}$ and $\sigma_{AB}$ are LU equivalent if and only if they are LP equivalent, or if and only if they are matched.
\end{proposition}

\begin{proof}
We first show that $\rho_{AB}$ and $\sigma_{AB}$ are LP equivalent if and only if they are matched. 
The ``Only if'' part is obvious. For the ``If'' part, assume that $\rho_{AB}$ and $\sigma_{AB}$ are matched, i.e., their LUBIs are all matched. First, by Lemma \ref{le:pcommute} we determine that $\sigma_{AB}$ is also a CLC state.

Second, due to the matched LUBIs formulated by $\tr[(\rho_A\ox I_B)^k]=\tr[(\sigma_A\ox I_B)^k]$ for all $k\geq 0$, it follows from the Newton identities that $\rho_A$ and $\sigma_A$ share the identical eigenvalues. Similarly, $\rho_B$ and $\sigma_B$ share the identical eigenvalues too. Then there exist permutations $P,Q$ acting on systems $A,B$ respectively, such that $\sigma_A=P\rho_A P^\dg$ and $\sigma_B=Q\rho_B Q^\dg$, as the marginals $\rho_A,\rho_B,\sigma_A,\sigma_B$ are all diagonal. Let $\tilde{\rho}_{AB}=(P\ox Q)\rho_{AB}(P\ox Q)^\dg$. Then up to a local permutation $P\ox Q$ we may assume that $\rho_{AB}$ and $\sigma_{AB}$ share the same two marginals, where $\rho_{AB}$ and $\sigma_{AB}$ remain both CLC states by Lemma \ref{le:com-inv}. 

Third, since the two marginals of $\rho_{AB}$ are both non-degenerate and $\sigma_{AB}$ shares the same marginals as $\rho_{AB}$, we may assume that $\rho_A=\sigma_A=\diag(a_1,\cdots,a_m)$ and $\rho_B=\sigma_B=\diag(b_1,\cdots,b_n)$, where $a_1,\cdots,a_m$ are distinct and $b_1,\cdots,b_n$ are distinct. It follows by Lemma \ref{le:diagonal} (i) that $\rho_{AB}$ and $\sigma_{AB}$ are both diagonal matrices. Then let $\rho_{AB}=\diag(\lambda_1,\lambda_2,\cdots,\lambda_{mn})$ and $\sigma_{AB}=\diag(\mu_1,\mu_2,\cdots,\mu_{mn})$. By calculation we obtain that
\beq
\label{eq:diag-m-3}
\bal
&\rho_A\ox\I_B=\sigma_A\ox\I_B \\ =&\diag(\overbrace{a_1,\cdots,a_1}^n,\overbrace{a_2,\cdots, a_2}^n,\cdots,\overbrace{a_m,\cdots,a_m}^n), \\
&\I_A\ox\rho_B=\I_A\ox\sigma_B \\
=&\diag(\underbrace{\underbrace{b_1,b_2,\cdots,b_n},\underbrace{b_1,b_2\cdots,b_n},\cdots,\underbrace{b_1,b_2,\cdots,b_n}}_m).
\eal
\eeq

Fourth, due to the local commutativity, the criteria given by Eq. \eqref{eq:cj-1} simplifies to 
\beq
\label{eq:cj-1-c}
\tr[\a_1^p\alpha_2^q\alpha_3^r]=\tr[\b_1^p\b_2^q\b_3^r]
\eeq
for any integers $p,q,r\geq 0$. Since $\alpha_2,\beta_2$ are two identical diagonal matrices, and $\alpha_3,\beta_3$ are two identical diagonal matrices, Eq. \eqref{eq:cj-1-c} tells that the following equality holds for any integers $p,q,r\geq 0$,
\beq
\label{eq:cj-1-c.1}
\bal
&\sum_{k=1}^m \sum_{l=1}^n (a_k)^q(b_l)^r(\lambda_{(k-1)n+l})^p \\
=&\sum_{k=1}^m \sum_{l=1}^n (a_k)^q(b_l)^r(\mu_{(k-1)n+l})^p.
\eal
\eeq
Fix $p=1$. Then the following equality,
\beq
\label{eq:cj-1-c.2}
\sum_{k=1}^m \sum_{l=1}^n (a_k)^q(b_l)^r\big(\lambda_{(k-1)n+l}-\mu_{(k-1)n+l}\big)=0,
\eeq
holds for any integers $q,r\geq 0$. 

Finally, we show that $\lambda_{(k-1)n+l}=\mu_{(k-1)n+l}$ for any $1\leq k\leq m$ and $1\leq l\leq n$. Eq. \eqref{eq:cj-1-c.2} can be taken as a system of linear equations, where the variables are $(\lambda_j-\mu_j)$'s. By selecting $q=0,1,\cdots,m-1$ and $r=0,1,\cdots,n-1$, a part of the system of linear equations can be formulated as
\beq
\label{eq:cj-1-c.3}
(M_m\ox N_n)\hat{X}=0,
\eeq
where
\begin{widetext}
\beq
\label{eq:cj-1-c.4}
M_m=
\bma
1 & 1 & \ldots & 1 \\
a_1 & a_2 & \ldots & a_m \\
\vdots & \vdots & \vdots & \vdots \\
a_1^{m-1} & a_2^{m-1} & \ldots & a_m^{m-1}
\ema,~
N_n=
\bma
1 & 1 & \ldots & 1 \\
b_1 & b_2 & \ldots & b_n \\
\vdots & \vdots & \vdots & \vdots \\
b_1^{n-1} & b_2^{n-1} & \ldots & b_n^{n-1}
\ema,~
\hat{X}=
\bma
\lambda_1-\mu_1 \\
\lambda_2-\mu_2 \\
\vdots \\
\lambda_{mn}-\mu_{mn}
\ema.
\eeq
\end{widetext}
It is known that $M_m\ox N_n$ is invertible if and only if $M_m$ and $N_n$ are both invertible. Thus, $\lambda_j=\mu_j$ holds for any $j$, if and only if $\det(M_m)$ and $\det(N_n)$ are both nonzero. By the assumption that $a_1,\cdots,a_m$ are distinct and $b_1,\cdots,b_n$ are distinct, the Vandermonde determinants $\det(M_m)$ and $\det(N_n)$ are both nonzero. Thus we derive that $\lambda_j=\mu_j$ for any $1\leq j\leq mn$. Recall that $\rho_{AB}=\diag(\lambda_1,\lambda_2,\cdots,\lambda_{mn})$ up to some local permutation, and $\sigma_{AB}=\diag(\mu_1,\mu_2,\cdots,\mu_{mn})$. We then conclude that $\rho_{AB}$ and $\sigma_{AB}$ are the same up to a local permutation. In other words, they are LP equivalent.

Next, we show that the LU equivalence is indeed the same as the LP equivalence for $\rho_{AB}$ and $\sigma_{AB}$. According to the analysis above we obtain that the two diagonal marginals of $\sigma_{AB}$ are also non-degenerate if it is matched with $\rho_{AB}$. Then by Lemma \ref{le:diagonal} (ii) we conclude that $\rho_{AB}$ and $\sigma_{AB}$ are LU equivalent if and only if they are LP equivalent when the two states are matched. 
This completes the proof.
\end{proof}

Based on the above conclusion on bipartite diagonal states, we give an explicit example, in terms of two CLC diagonal states, illustrating the additional discriminatory power of local-unitary Bargmann invariants over several commonly used necessary invariants for local-unitary equivalence, including the global spectrum, the spectra of the reduced states, and the operator-Schmidt coefficients. More precisely, we consider the following two diagonal states on $\mathbb C^3\otimes\mathbb C^3$.

\begin{example}
\label{ex:adv}
Consider the following two states,
\beq
\label{eq:exadv-1}
\bal
\rho_{AB}&=\frac{1}{45}\diag (1,2,3,4,6,5,8,7,9), \\
\sigma_{AB}&=\frac{1}{45}\diag (1,2,3,5,4,6,7,9,8)
\eal
\eeq
with the same global spectra. Their marginals
\beq
\label{eq:exadv-2}
\bal
\rho_A&=\sigma_A=\diag (\frac{2}{15},\frac{1}{3},\frac{8}{15}),\\
\rho_B&=\sigma_B=\diag (\frac{13}{45},\frac{15}{45},\frac{17}{45})
\eal
\eeq
are identical and non-degenerate, and with the same operator-Schmidt coefficients. 

Since $\rho_{AB}$ and $\sigma_{AB}$ are diagonal and thus necessarily CLC states, it follows from Proposition \ref{thm:distinct} that $\rho_{AB}$ and $\sigma_{AB}$ are LU equivalent if and only if all the LUBIs are matched, as all the marginals of $\rho_{AB}$ and $\sigma_{AB}$ are non-degenerate. Here we show that there is a fourth-order LUBI that is distinct for $\rho_{AB}$ and $\sigma_{AB}$. We have the following two LUBIs associated with the same string for $\rho_{AB}$ and $\sigma_{AB}$,
\beq
\label{eq:exadv-3}
\bal
\Delta_{1233}(\rho)=\tr[\rho_{AB}(\rho_A\ox\I_B)(\I_A\ox\rho_B)^2]=\frac{193653}{45^4},\\
\Delta_{1233}(\sigma)=\tr[\sigma_{AB}(\sigma_A\ox\I_B)(\I_A\ox\sigma_B)^2]=\frac{193581}{45^4}.
\eal
\eeq
These two LUBIs are distinct, and thus $\rho_{AB}$ and $\sigma_{AB}$ are not LU equivalent by Proposition \ref{thm:distinct}.
\end{example}

The above example demonstrates the discriminatory power of local-unitary Bargmann invariants relative to the common standard necessary invariants. We may summarize the example as follows:
\[
\begin{array}{c|c}
\text{quantity} & \rho_{AB}\ \text{and}\ \sigma_{AB} \\
\hline
\text{global spectrum} & \text{identical} \\
\rho_A,\ \sigma_A & \text{identical and nondegenerate} \\
\rho_B,\ \sigma_B & \text{identical and nondegenerate} \\
\text{operator-Schmidt coefficients} & \text{identical} \\
\text{Bargmann invariants of order }\leq 3 & \text{identical} \\
\Delta_{1233} & \text{different}
\end{array}
\]

Based on the above proposition and specific example, we remark that it suffices to employ finite matched LUBIs to completely characterize the LU orbits of some bipartite CLC states with two diagonal non-degenerate marginals. In this case, the criteria based on the matched LUBIs can be simplified to Eq. \eqref{eq:cj-1-c}. According to the analysis below Eq. \eqref{eq:cj-1-c}, by selecting $p=1$, $0\leq q\leq m-1$ and $0\leq r\leq n-1$, the corresponding matched LUBIs are sufficient to determine whether $\rho_{AB}$ and $\sigma_{AB}$ are LU equivalent or not. By applying proper LU operators we generalize the above result to the bipartite CLC states whose marginals may be not diagonal in the following theorem.

\begin{theorem}
\label{thm:comm+sim}
Let $\rho_{AB}$ be a CLC state and its two marginals both are non-degenerate. For some state $\sigma_{AB}$, $\rho_{AB}$ and $\sigma_{AB}$ are LU equivalent if and only if the two states are matched.
\end{theorem}

\begin{proof}
The ``Only if'' part is obvious. We show the ``If'' part. Assume that $\rho_{AB}$ and $\sigma_{AB}$ are matched, i.e., their LUBIs are all matched. By Lemma \ref{le:pcommute} the two bipartite states are both CLC states if one of the two is a CLC state and they share all the matched LUBIs. Up to LU equivalence we may assume that the marginals of $\tilde{\rho}_{AB}$ and $\tilde{\sigma}_{AB}$ all are diagonal and non-degenerate, where $\tilde{\rho}_{AB}$ and $\tilde{\sigma}_{AB}$ are LU equivalent to $\rho_{AB}$ and $\sigma_{AB}$, respectively. Also, $\tilde{\rho}_{AB}$ and $\tilde{\sigma}_{AB}$ are still CLC states according to Lemma \ref{le:com-inv}. Therefore, from Proposition \ref{thm:distinct}, finite matched LUBIs can ensure that $\tilde{\rho}_{AB}$ and $\tilde{\sigma}_{AB}$ are LP equivalent. It implies that $\rho_{AB}$ and $\sigma_{AB}$ are LU equivalent. This completes the proof.
\end{proof}

The local commutativity defined by Def. \ref{def:commute} can be directly extended to general multipartite Hermitian operators. The entanglement witnesses are special Hermitian operators, detecting effectively the entanglement of unknown quantum states. Thus, it is necessary to classify bipartite/multipartite EWs under LU equivalence for characterizing entanglement \cite{inertiays2020,rew2024,yiprr2025}. Actually the classification of EWs is closely connected to that of states. Inspired by the idea of comparing the LU decidabilities between bipartite states and EWs \cite{yiprr2025}, we may use the matched LUBIs correspondingly generated from the EWs to test the LU equivalence between bipartite EWs.

\begin{proposition}
\label{prop:biew}
Suppose that $W_{AB}$ is an EW supported on $\bbC^m\ox\bbC^n$ whose marginals are both non-degenerate, commuting with both $W_A\ox\I_B$ and $\I_A\ox W_B$, where $W_A=\tr_B W_{AB}$ and $W_B=\tr_A W_{AB}$. For another EW $X_{AB}$, the LU equivalence between $W_{AB}$ and $X_{AB}$ is equivalent to that between the two states $\rho_{AB}=\frac{1}{1+cmn}(W_{AB}+c\I_{mn})$ and $\sigma_{AB}=\frac{1}{1+cmn}(X_{AB}+c\I_{mn})$ for some $c>0$, whose LU equivalence can be verified by the matched LUBIs similarly given by Eq. \eqref{eq:cj-1}.
\end{proposition}

\begin{proof}
First, there exists some $c>0$ such that both $W_{AB}+c\I_{mn}$ and $X_{AB}+c\I_{mn}$ are positive semidefinite, and thus $\rho_{AB}=\frac{1}{1+cmn}(W_{AB}+c\I_{mn})$ and $\sigma_{AB}=\frac{1}{1+cmn}(X_{AB}+c\I_{mn})$ are two normalized states. Then we conclude that the LU equivalence between two EWs $W_{AB}$ and $X_{AB}$ is equivalent to that between two states $\rho_{AB}$ and $\sigma_{AB}$. It is verified that that $\rho_{AB}$ is a CLC state if and only if $W_{AB}$ commutes with both $W_A\ox\I_B$ and $\I_A\ox W_B$, and the marginals of $\rho_{AB}$ are non-degenerate if and only if $W_A,W_B$ are correspondingly non-degenerate. Thus, by assumption the state $\rho_{AB}$ is a CLC state with two non-degenerate marginals. According to Theorem \ref{thm:comm+sim}, the LU equivalence between $\rho_{AB}$ and $\sigma_{AB}$ is verified by finite matched LUBIs, which is the same as the LU equivalence between $W_{AB}$ and $X_{AB}$. This completes the proof.
\end{proof}

\subsection{Identifying multipartite LU equivalence}
\label{subsec:multi-lu}

We use the matched LUBIs formulated in \eqref{eq:cj-2m} to further test the LU equivalence between multipartite states. We first consider the tripartite case and show that the matched LUBIs are sufficient to verify the LU equivalence between tripartite CLC states with non-degenerate single-party marginals, and the show that this assertion can be directly generalized to the $n$-partite states for any $n\geq 3$. We first characterize the multipartite CLC states with diagonal single-party marginals, by generalizing Lemma \ref{le:diagonal} to the multipartite case.

\begin{lemma}
\label{le:tri-diag}
If $X_{A_1\cdots A_n}$ is an $n$-partite Hermitian operator which commutes with $Y_j\ox \I_{A_{j^c}}$ for any $1\leq j\leq n$, where $Y_j$ is a diagonal and non-degenerate operator acting on the system $A_j$, then $X_{A_1\cdots A_n}$ must be diagonal.
\end{lemma}

\begin{proof}
We show that Lemma \ref{le:tri-diag} holds for any $n\geq2$ by mathematical induction. First, the assertion holds for $n=2$ by Lemma \ref{le:diagonal} (i). Second, we assume that the assertion holds for any $(n-1)$-partite system with $n\geq 3$. Finally, we show the assertion holds for the $n$-partite system by the above assumption. Suppose that $X_{A_1\cdots A_n}$ is supported on $\bigox_{i=1}^n\bbC^{d_i}$, where $d_i$ denotes the dimension of subsystem $A_i$. Then we may write $X_{A_1\cdots A_n}$ in a form of block matrix,
\beq
\label{eq:tridiag-m-1}
X_{A_1\cdots A_n}=
\bma
A_{11} & A_{12} & \cdots & A_{1d_1} \\
A_{12}^\dg & A_{22} & \cdots & A_{2d_1} \\
\vdots & \vdots & \cdots & \vdots \\
A_{1d_1}^\dg & A_{2d_1}^\dg & \cdots & A_{d_1d_1}
\ema,
\eeq
where each $A_{ij}$ is an $(n-1)$-partite matrix supported on $\bigox_{i=2}^{n}\bbC^{d_i}$. Under the bipartition $A_1|(A_2\cdots A_n)$, it follows from $[X_{A_1\cdots A_n},Y_1\ox \I_{A_2\cdots A_n}]=0$ that the non-diagonal blocks in \eqref{eq:tridiag-m-1} are all zero blocks, as $Y_1$ is a diagonal matrix with distinct diagonal entries. Further, since $X_{A_1\cdots A_n}$ commutes with $Y_k\ox \I_{A_{k^c}}$ for each $2\leq k\leq n$, it implies that each diagonal block $A_{jj}$ commutes with every $Y_k\ox\I_{A_{\{1,k\}^c}}$ for $2\leq k \leq n$, where $A_{\{1,k\}^c}$ denotes an $(n-1)$-partite system without subsystems $A_1$ and $A_k$. According to the assumption for the $(n-1)$-partite system, we conclude that each $(n-1)$-partite operator $A_{jj}$ must be diagonal, as every $Y_j$ is diagonal and non-degenerate. It means that $X_{A_1\cdots A_n}$ is diagonal and thus the assertion holds for the $n$-partite system. By mathematical induction we conclude that Lemma \ref{le:tri-diag} holds for any $n\geq 2$. This completes the proof.
\end{proof}

Next, we take the tripartite case as an example to illustrate the results on multipartite systems. We generalize the Proposition \ref{thm:distinct} to tripartite systems.

\begin{proposition}
    \label{le:tri-lp}
    Suppose that $\rho_{ABC}$ is a CLC state with three single-party marginals that are all diagonal and non-degenerate. For some state $\sigma_{ABC}$ with three single-party marginals that are all diagonal, $\sigma_{ABC}$ are LP equivalent to $\rho_{ABC}$ if and only if they are matched.
\end{proposition}

\begin{proof}
The ``Only if'' part is direct. We show the ``If'' part. Assume that $\rho_{ABC}$ and $\sigma_{ABC}$ supported on $\bbC^m\ox\bbC^n\ox\bbC^l$ share all the matched LUBIs. Since $\rho_{ABC}$ is a CLC state, we conclude that $\sigma_{ABC}$ is also a CLC state by Lemma \ref{le:pcommute}.

Similar to the bipartite case, we obtain that $\rho_A,\rho_B,\rho_C$ share the identical eigenvalues as $\sigma_A,\sigma_B,\sigma_C$, respectively, due to the matched LUBIs like $\tr[(\rho_A\ox\I_{BC})^k]=\tr[(\sigma_A\ox\I_{BC})^k]$ for all $k\geq 0$. Then up to a tripartite local permutation we may assume that $\rho_{ABC}$ share the same three single-party marginals as $\sigma_{ABC}$, namely,
    \beq
    \label{eq:tri-lp-1}
    \bal
    \rho_A&=\sigma_A=\diag(a_1,\cdots,a_m), \\ \rho_B&=\sigma_B=\diag(b_1,\cdots,b_n), \\
   \rho_C&=\sigma_C=\diag(c_1,\cdots,c_l),
    \eal
    \eeq
where the diagonal entries in each diagonal matrix are distinct.

It suffices to consider the following matched LUBIs given in \eqref{eq:cj-2m},
\beq
\label{eq:tri-lp-2}
\tr[\rho_{ABC}(\rho_A^p\ox\rho_B^q\ox\rho_C^r)]=\tr[\sigma_{ABC}(\sigma_A^p\ox\sigma_B^q\ox\sigma_C^r)].
\eeq
By the condition of commutativity, it follows from Lemma \ref{le:tri-diag} that $\rho_{ABC}$ and $\sigma_{ABC}$ are both diagonal. Let $\rho_{ABC}=\diag(\lambda_1,\cdots,\lambda_{mnl})$ and $\sigma_{ABC}=\diag(\mu_1,\cdots,\mu_{mnl})$. By fixing $p=0,\cdots,m-1$, $q=0,\cdots,n-1$ and $r=0,\cdots,l-1$, the equalities, $\tr[\rho_{ABC}(\rho_A^p\ox\rho_B^q\ox\rho_C^r)]-\tr[\sigma_{ABC}(\sigma_A^p\ox\sigma_B^q\ox\sigma_C^r)]=0$, make up the system of linear equations formulated by
\beq
\label{eq:tri-lp-3}
(M_m\ox N_n\ox L_l) \hat{X}=0,
\eeq
where
\beq
\label{eq:tri-lp-3.1}
\bal
M_m&=
\bma
1 & 1 & \ldots & 1 \\
a_1 & a_2 & \ldots & a_m \\
\vdots & \vdots & \vdots & \vdots \\
a_1^{m-1} & a_2^{m-1} & \ldots & a_m^{m-1}
\ema,\\
N_n&=
\bma
1 & 1 & \ldots & 1 \\
b_1 & b_2 & \ldots & b_n \\
\vdots & \vdots & \vdots & \vdots \\
b_1^{n-1} & b_2^{n-1} & \ldots & b_n^{n-1}
\ema, \\
L_l&=
\bma
1 & 1 & \ldots & 1 \\
c_1 & c_2 & \ldots & c_l \\
\vdots & \vdots & \vdots & \vdots \\
c_1^{l-1} & c_2^{l-1} & \ldots & c_l^{l-1}
\ema, \\
\hat{X}&=
\bma
\lambda_1-\mu_1 \\
\lambda_2-\mu_2 \\
\vdots \\
\lambda_{mnl}-\mu_{mnl}
\ema.
\eal
\eeq
Since the tuples $(a_1,\cdots,a_m),(\b_1,\cdots,b_n),(c_1,\cdots,c_l)$ respectively have distinct entries, we obtain that $\hat{X}$ must be zero from $\det(M_m\ox N_n\ox L_l)\neq 0$. Thus, we derive that $\rho_{ABC}=\sigma_{ABC}$. Taking the local permutation into account, we conclude that $\rho_{ABC}$ is LP equivalent to $\sigma_{ABC}$.
\end{proof}

Based on Proposition \ref{le:tri-lp} we generalize the LP equivalence to the LU equivalence by extending the diagonal marginals to general marginals.


\begin{theorem}
\label{thm:tri-lu}
Suppose that $\rho_{ABC}$ is a CLC state with three single-party marginals that are all non-degenerate. A state $\sigma_{ABC}$ is LU equivalent to $\rho_{ABC}$ if and only if the two states are matched.
\end{theorem}

\begin{proof}
The ``Only if'' part is direct. We show the ``If'' part. Since $\sigma_{ABC}$ has the matched LUBIs with $\rho_{ABC}$, it follows by Lemma \ref{le:pcommute} that $\sigma_{ABC}$ is also a CLC state.
Up to the LU equivalence we may assume that the single-party marginals of $\tilde{\rho}_{ABC}$ and $\tilde{\sigma}_{ABC}$ all are diagonal and non-degenerate, where $\tilde{\rho}_{ABC}$ and $\tilde{\sigma}_{ABC}$ are LU equivalent to $\rho_{ABC}$ and $\sigma_{ABC}$ respectively for some tripartite LU operators. $\tilde{\rho}_{ABC}$ and $\tilde{\sigma}_{ABC}$ are also CLC states according to Lemma \ref{le:com-inv}. Therefore, based on Proposition \ref{le:tri-lp}, the matched LUBIs given by Eq. \eqref{eq:cj-1} ensure that $\tilde{\rho}_{ABC}$ and $\tilde{\sigma}_{ABC}$ are the same up to a local permutation. That is, there exists a tripartite local permutation $P$ such that $\tilde{\sigma}_{ABC}=P\tilde{\rho}_{ABC} P^\dg$. It follows that $\rho_{ABC}$ and $\sigma_{ABC}$ are LU equivalent.
\end{proof}

The study on the tripartite case can be directly extended to the $n$-partite case for any $n>3$. We present the general conclusion on the $n$-partite case as follows.

\begin{corollary}
\label{cor:lu-multi}
Suppose that the single-party marginals of the two states $\rho,\sigma$ of system $A_1\cdots A_n$ are all non-degenerate. If one of the two states is a CLC state, then the two $n$-partite states $\rho,\sigma$ are LU equivalent if and only if they are matched.
\end{corollary}

One may derive Corollary \ref{cor:lu-multi} in a similar way to the tripartite case, by establishing the LP equivalence for the $n$-partite states with diagonal single-party marginals, followed by the LU equivalence for the $n$-partite states with general single-party marginals.

\section{Counterexamples in the degenerate case}
\label{sec:cex}

In this section we show that all the matched LUBIs cannot ensure the LU equivalence between two general states. We consider the case where the states have degenerate single-party marginals, and propose specific pairs of states, where the states in each pair are matched but are not LU equivalent.

To test the LU equivalence of the states with degenerate single-party marginals, we begin with the extreme case where each single-party marginal is proportional to the identity. With this condition, the criteria of LU equivalence formulated by the matched LUBIs are reduced to the following equalities,
\beq
\label{eq:cj-2s}
\tr[\rho^k]=\tr[\sigma^k],~~\forall k>0,
\eeq 
where $\rho$ and $\sigma$ are two multipartite states that are matched. We show in this extreme case that there exist two states which are matched but are not LU equivalent by the following examples.

\begin{example}
\label{ex:c-ex}
Let us consider two states in spectral decompositions below,
\beq
\label{eq:cex-2}
\bal
\rho_{AB}&=p\ketbra{a}+q\ketbra{b}+r\ketbra{c}, \\ 
\sigma_{AB}&=p\ketbra{a}+q\ketbra{d}+r\ketbra{f}
\eal
\eeq
for distinct $p,q,r >0$ with $p+q+r=1$,
where $\omega=e^{\frac{2\pi}{3}i}$,
\beq
\label{eq:cex-1}
\bal
\ket{a}&=\frac{1}{\sqrt{3}}(\ket{00}+\ket{11}+\ket{22}), \\
\ket{b}&=\frac{1}{\sqrt{3}}(\ket{00}+\omega\ket{11}+\omega^2\ket{22}), \\
\ket{c}&=\frac{1}{\sqrt{3}}(\ket{00}+\omega^2\ket{11}+\omega\ket{22}), \\
\ket{d}&=\frac{1}{\sqrt{3}}(\ket{01}+\ket{12}+\ket{20}), \\
\ket{f}&=\frac{1}{\sqrt{3}}(\ket{01}+\omega\ket{12}+\omega^2\ket{20}). \\
\eal
\eeq
 We claim that $\rho_{AB}$ is not LU equivalent to $\sigma_{AB}$, while the two states $\rho_{AB}$ and $\sigma_{AB}$ are matched. 
\end{example}

\begin{proof}
First, one can verify that $\rho_A=\rho_B=\sigma_A=\sigma_B=\frac{1}{3}I_3$. It follows that for each pair $(j,k)$, $\alpha_j$ and $\alpha_k$ are commutative, and $\b_j$ and $\beta_k$ are commutative. Hence, it is equivalent to considering Eq. \eqref{eq:cj-2s} for any positive integer $k$. Second, the five pure states given by Eq. \eqref{eq:cex-1} are all maximally entangled, and are mutually orthogonal. Then the spectral decompositions of $\rho_{AB}$ and $\sigma_{AB}$ read as Eq. \eqref{eq:cex-2}. Thus, $\tr[\rho^k_{AB}]=\tr[(\sigma_{AB})^k]$ holds for any $k>0$, as $\rho_{AB}$ and $\sigma_{AB}$ share identical eigenvalues. However, we claim that $\rho_{AB}$ and $\sigma_{AB}$ given by Eq. \eqref{eq:cex-2} are not LU equivalent.

We show the claim above by contradiction. Assume that $\rho_{AB}$ and $\sigma_{AB}$ are LU equivalent. By Lemma \ref{le:lu-equiv-g}, it means that there exists a common local unitary $U\ox V$ such that the following equalities hold,
\begin{eqnarray}
\label{eq:cex-3.1}
&&(U\ox V)\ket{a}=e^{i\a}\ket{a}, \\
\label{eq:cex-3.2}
&&(U\ox V)\ket{b}=e^{i\b}\ket{d}, \\
\label{eq:cex-3.3}
&&(U\ox V)\ket{c}=e^{i\gamma}\ket{f}.
\end{eqnarray}
By calculation we obtain that 
\beq
\label{eq:cex-4}
\left\{
\bal
V&=e^{i\theta} U^*, \\
V
\bma
1 & 0 & 0 \\
0 & \omega & 0 \\
0 & 0 & \omega^2 
\ema
V^\dg &\propto
\bma
0 & 1 & 0 \\
0 & 0 & 1 \\
1 & 0 & 0
\ema, \\
V
\bma
1 & 0 & 0 \\
0 & \omega^2 & 0 \\
0 & 0 & \omega 
\ema
V^\dg &\propto
\bma
0 & 1 & 0 \\
0 & 0 & \omega \\
\omega^2 & 0 & 0
\ema,
\eal
\right.
\eeq
where $U^*$ is the complex conjugate of $U$.
Due to $V
\bma
1 & 0 & 0 \\
0 & \omega & 0 \\
0 & 0 & \omega^2 
\ema
V^\dg V
\bma
1 & 0 & 0 \\
0 & \omega^2 & 0 \\
0 & 0 & \omega 
\ema
V^\dg=I_3$, from Eq. \eqref{eq:cex-4} we obtain a contradiction as $\bma
0 & 1 & 0 \\
0 & 0 & 1 \\
1 & 0 & 0
\ema\cdot
\bma
0 & 1 & 0 \\
0 & 0 & \omega \\
\omega^2 & 0 & 0
\ema$ cannot be proportional to $I_3$. It implies that such unitary $V$ does not exist. Therefore, $\rho_{AB}$ and $\sigma_{AB}$ are not LU equivalent.
\end{proof} 

One can alternatively verify that the two states $\rho_{AB}$ and $\sigma_{AB}$ proposed in Example \ref{ex:c-ex} are not LU equivalent, in terms of the criterion given by \cite[Eq. (7)]{bilu2012}. We show the alternative verification in Appendix \ref{sec:verify2}.
Example \ref{ex:c-ex} can be generalized to high-dimensional bipartite systems with equal local dimensions. We propose two states supported on $\bbC^d\ox\bbC^d$ ($d\geq 3$) in the following example which are matched but are not LU equivalent.

\begin{example}
\label{ex:cex-2}
Define the following states in $\bbC^d\ox\bbC^d$,
\beq
\label{eq:cex-1.1}
\bal
\ket{a_0}&=\frac{1}{\sqrt{d}}\sum_{j=0}^{d-1}\ket{j,j}, \\
\ket{a_1}&=\frac{1}{\sqrt{d}}\sum_{j=0}^{d-1}\omega^j \ket{j,j}, \\
\ket{a_2}&=\frac{1}{\sqrt{d}}\big(\sum_{j=1}^{d-1}\omega^{\tau_1(j)} \ket{j,j}\big), \\
&\cdots \\
\ket{a_{d-1}}&=\frac{1}{\sqrt{d}}\sum_{j=0}^{d-1}\omega^{\tau_{d-2}(j)} \ket{j,j}, \\
\ket{b_1}&=\frac{1}{\sqrt{d}}\sum_{j=0}^{d-1}\ket{j,j+1\mod d}, \\
\ket{b_2}&=\frac{1}{\sqrt{d}}\sum_{j=0}^{d-1}\omega^j\ket{j,j+1\mod d}, \\
\ket{b_3}&=\frac{1}{\sqrt{d}}\sum_{j=0}^{d-1}\omega^{\tau_1(j)}\ket{j,j+1\mod d}, \\
&\cdots \\
\ket{b_{d-1}}&=\frac{1}{\sqrt{d}}\sum_{j=0}^{d-1}\omega^{\tau_{d-3}(j)}\ket{j,j+1\mod d}, \\
\eal
\eeq
where $\omega=e^{\frac{2\pi}{d}i}$, and for each $1\leq k \leq d-2$, $\tau_k$ is a permutation of $\{1,\cdots,d-1\}$ defined by
\beq
\label{eq:cex-1.2}
\tau_k=\underbrace{(1,2,3,\cdots,d-1)\cdots(1,2,3,\cdots,d-1)}_{k}
\eeq
which means $k$ times of the cycle of $\{1,\cdots,d-1\}$. 
Consider the following states,
\beq
\label{eq:cex-1.3}
\bal
\rho_{AB}&=\sum_{j=0}^{d-1} p_j\ketbra{a_j}, \\
\sigma_{AB}&=p_0\ketbra{a_0}+\sum_{j=1}^{d-1} p_j\ketbra{b_j},
\eal
\eeq
where $p_j$'s are all distinct and $\sum_{j=0}^{d-1}=1$.
We claim that $\rho_{AB}$ is not LU equivalent to $\sigma_{AB}$, while the two states $\rho_{AB}$ and $\sigma_{AB}$ are matched. 
\end{example}

\begin{proof}
First, it is verified that the vectors in $\{\ket{a_0},\ket{a_1},\cdots,\ket{a_{d-1}},\ket{b_1},\cdots,\ket{b_{d-1}}\}$ are pairwisely orthonormal. Thus, the decompositions of $\rho_{AB}$ and $\sigma_{AB}$ given by Eq. \eqref{eq:cex-1.3} are spectral decompositions. Second, one can show that $\rho_A=\rho_B=\sigma_A=\sigma_B=\frac{1}{d}\I_d$. It follows that for each pair $(j,k)$, $\alpha_j$ and $\alpha_k$ are commutative, and $\b_j$ and $\beta_k$ are commutative. Hence, it is equivalent to considering Eq. \eqref{eq:cj-2s} for any positive integer $k$. Since $\rho_{AB}$ and $\sigma_{AB}$ share identical eigenvalues, Eq. \eqref{eq:cj-2s}, namely, $\tr[\rho^k_{AB}]=\tr[(\sigma_{AB})^k]$, always holds for any positive integer $k$. However, we claim that $\rho_{AB}$ and $\sigma_{AB}$ given by Eq. \eqref{eq:cex-2} are not LU equivalent.

We show the above claim by contradiction. Assume that $\rho_{AB}$ and $\sigma_{AB}$ are LU equivalent. By Lemma \ref{le:lu-equiv-g}, it means there exists a common LU operator $U\ox V$ such that the following equalities hold,
\begin{eqnarray}
\label{eq:cex-1.4-1}
&&(U\ox V) \ket{a_0}=e^{i\alpha_0}\ket{a_0}, \\
\label{eq:cex-1.4-2}
&&(U\ox V) \ket{a_j}=e^{i\alpha_j}\ket{b_j},~~\forall 1\leq j\leq d-1.
\end{eqnarray}
By calculation we obtain that $V=e^{i\theta} U^*$ from Eq. \eqref{eq:cex-1.4-1}, and obtain from Eq. \eqref{eq:cex-1.4-2} that
\begin{widetext}
\beq
\label{eq:cex-1.6}
\left\{
\bal
V
\left(\sum_{j=0}^{d-1}\omega^{j}\ketbra{j}\right)
V^\dg &\propto
\sum_{j=0}^{d-2}\ketbra{j}{j+1}+\ketbra{d-1}{0}, \\
V
\left(\ketbra{0}+\sum_{j=1}^{d-1}\omega^{\tau_1(j)}\ketbra{j}\right)
V^\dg &\propto
\ketbra{0}{1}+\left(\sum_{j=1}^{d-2}\omega^j\ketbra{j}{j+1}\right)+\omega^{d-1}\ketbra{d-1}{0}, \\
V
\left(\ketbra{0}+\sum_{j=1}^{d-1}\omega^{\tau_2(j)}\ketbra{j}\right)
V^\dg &\propto
\ketbra{0}{1}+\left(\sum_{j=1}^{d-2}\omega^{\tau_1(j)}\ketbra{j}{j+1}\right)+\omega^{\tau_1(d-1)}\ketbra{d-1}{0}, \\
&\ldots, \\
V
\left(\ketbra{0}+\sum_{j=1}^{d-1}\omega^{\tau_{d-2}(j)}\ketbra{j}\right)
V^\dg &\propto
\ketbra{0}{1}+\left(\sum_{j=1}^{d-2}\omega^{\tau_{d-3}(j)}\ketbra{j}{j+1}\right)+\omega^{\tau_{d-3}(d-1)}\ketbra{d-1}{0}, \\
\eal
\right.
\eeq
\end{widetext}
where $U^*$ is the complex conjugate of $U$. One can verify that the multiplication of each left hand side of Eq. \eqref{eq:cex-1.6} is equal to $\I_d$. However, by direct calculation the multiplication of each right hand side of Eq. \eqref{eq:cex-1.6} is not proportional to $\I_d$. It implies that such unitary $V$ does not exist. Therefore, $\rho_{AB}$ and $\sigma_{AB}$ are not LU equivalent.
\end{proof}

Next, we show that all the matched LUBIs are not sufficient to enusre the LU equivalence between multipartite states by the following example.
Similar to the idea of dealing with the case of bipartite states, we further construct two three-qubit states in Example \ref{ex:cex-tri} whose marginals are all proportional to identity matrices in any bipartition of the three-qubit system.

\begin{example}
\label{ex:cex-tri}
We construct two three-qubit states as follows,
\begin{widetext}
\beq
\label{eq:cex-tri-1}
\bal
\rho&=\frac{1}{4}(\ketbra{000}+\ketbra{011}+\ketbra{101}+\ketbra{110}), \\
\sigma(\theta)&=\frac{1}{8}\Big[(\ket{000}+\ket{111})(\bra{000}+\bra{111})+(\ket{001}+\ket{110})(\bra{001}+\bra{110}) \\
&~~~+(\ket{010}+\ket{101})(\bra{010}+\bra{101})+(e^{i\theta}\ket{011}+\ket{100})(e^{-i\theta}\bra{011}+\bra{100})\Big],
\eal
\eeq
\end{widetext}
where $\theta\in(0,2\pi)$. For any $\theta\in(0,2\pi)$, the two states $\rho$ and $\sigma(\theta)$ are matched but are not LU equivalent.
\end{example}

\begin{proof}
First, one can verify that for either $\rho$ or $\sigma(\theta)$, and each bipartition of the three-qubit system, the two marginals always are $\frac{1}{2}I_2$ and $\frac{1}{4}I_4$, respectively. It follows that for any pair $(j,k)$, $\alpha_j$ and $\alpha_k$ given by Eq. \eqref{eq:cj-2m.1} are commutative, and similarly $\alpha'_j$ and $\alpha'_k$ are also commutative. Then for the pair of states proposed in Eq. \eqref{eq:cex-tri-1}, according to the equalities of matched LUBIs formulated by Eq. \eqref{eq:cj-2m} it is equivalent to verifying that 
\beq
\label{eq:ctri-1}
\tr[\rho^k]=\tr[(\sigma(\theta))^k]
\eeq
for any integer $k>0$. Since $\rho$ and $\sigma(\theta)$ have identical eigenvalues, we conclude that Eq. \eqref{eq:ctri-1} holds for any integer $k>0$. Thus, $\rho$ and $\sigma(\theta)$ given by Eq. \eqref{eq:cex-tri-1} share all the matched LUBIs.

Second we show that $\rho$ and $\sigma(\theta)$ for $\theta\in(0,2\pi)$ cannot be LU equivalent. Consider the two ranges of $\rho$ and $\sigma(\theta)$, denoted by $\cR(\rho)$ and $\cR(\sigma(\theta))$. Obviously, the $\cR(\rho)$ is spanned by the four product vectors $\ket{000},\ket{011},\ket{101},\ket{101}$. Here we claim that $\cR(\sigma(\theta))$ contains no product vector when $\theta\in(0,2\pi)$, and thus $\rho$ and $\sigma(\theta)$ are not LU equivalent. We show this claim by contradiction. Assume that there is product vector in $\cR(\sigma(\theta))$. It follows that there exist coefficients $c_1,\cdots,c_4$ that are not all zero, such that
\beq
\label{eq:ctri-2}
\bal
v&=c_1(\ket{000}+\ket{111})+c_2(\ket{001}+\ket{110}) \\
&+c_3(\ket{010}+\ket{101})+c_4(e^{i\theta}\ket{011}+\ket{100})
\eal
\eeq
is a product vector. We may rewrite $v$ as
\beq
\label{eq:ctri-2.1}
\bal
v&=\ket{0}(c_1\ket{00}+c_2\ket{01}+c_3\ket{10}+c_4e^{i\theta}\ket{11}) \\ 
&+\ket{1}(c_4\ket{00}+c_3\ket{01}+c_2\ket{10}+c_1\ket{11}).
\eal
\eeq
When $v$ is a product vector, the two vectors $c_1\ket{00}+c_2\ket{01}+c_3\ket{10}+c_4e^{i\theta}\ket{11}$ and $c_4\ket{00}+c_3\ket{01}+c_2\ket{10}+c_1\ket{11}$ have to be proportional. It means that 
\beq
\label{eq:ctri-2.2}
\frac{c_1}{c_4}=\frac{c_4e^{i\theta}}{c_1}=\frac{c_2}{c_3}=\frac{c_3}{c_2}=t.
\eeq
We can rule out the case when $c_1=c_4=0$. Otherwise, $c_2\ket{01}+c_3\ket{10}$ cannot be a bipartite product vector for $c_2c_3\neq 0$. Similarly, we can also rule out the case when $c_2=c_3=0$. Thus, we obtain $c_1c_2c_3c_4\neq 0$. One can verify that the proportionality coefficient $t$ in Eq. \eqref{eq:ctri-2.2} can only be $\pm 1$. By Eq. \eqref{eq:ctri-2.2} we obtain that $e^{i\theta}=1$. However, $e^{i\theta}\neq 1$ for $\theta\in(0,2\pi)$. Hence, we obtain a contradiction. In other words, for $\theta\in(0,2\pi)$ the vector $v$ given by Eq. \eqref{eq:ctri-2} cannot be a product one. Thus, the range of $\sigma(\theta)$ contains no product vector. In light of this, $\rho$ and $\sigma(\theta)$ for $\theta\in(0,2\pi)$ proposed by Example \ref{ex:cex-tri} are not LU equivalent.
\end{proof}

\section{conclusion}
\label{sec:con}

We have investigated a criterion of LU equivalence for bipartite and multipartite states by the LUBIs. All the associated LUBIs from two respective states formulated in Definition \ref{def:match} are necessarily matched, if they are LU equivalent. Then we have focused on the converse, i.e., whether the matching of all LUBIs is sufficient to ensure the LU equivalence. We tested the special case involving the so-called local commutativity, i.e., the multipartite state $\rho_{A_1\cdots A_n}$ commutes with $\rho_{A_{\cS}}\ox\I_{A_{\cS^c}}$ for some subsystem $A_{\cS}$. If the state commutes with every $\rho_{A_{\cS}}\ox\I_{A_{\cS^c}}$ with $\abs{\cS}=1$, the state is a CLC one by Def. \ref{def:commute}. Take the bipartite case as an example. The criterion in terms of all matched LUBIs can be reduced to the standard form consisting of the equalities given by Eq. \eqref{eq:bi-sta} for the bipartite CLC states. With such a simplified criterion, we have shown by Proposition \ref{thm:distinct} that for a CLC state $\rho_{AB}$ with two non-degenerate diagonal marginals, (i) the state that is LU equivalent to $\rho_{AB}$ is indeed LP equivalent to $\rho_{AB}$, and (ii) the state $\sigma_{AB}$ is LU equivalent to $\rho_{AB}$ if and only if they are matched. It follows by Lemma \ref{le:com-inv} that the property of local commutativity remains unchanged after applying LU operations. Based on this fact we extended Proposition \ref{thm:distinct} to the states whose marginals may not be diagonal. We have shown by Theorem \ref{thm:comm+sim} that for a CLC state $\rho_{AB}$ with two non-degenerate marginals, another state $\sigma_{AB}$ is LU equivalent to $\rho_{AB}$ if and only if they are matched. Then we generalized the above two results to the case of multipartite CLC states by Proposition \ref{le:tri-lp} and Theorem \ref{thm:tri-lu}. Specifically, for two multipartite CLC states with non-degenerate single-party marginals, all matched LUBIs are sufficient to ensure their LU equivalence. It is worth mentioning that at most $D$ (the global dimension) LUBIs are needed to verify the LU equivalence between such states. It is indicated that the condition of non-degenerate marginals is crucial. We have shown by examples that the criterion in terms of all matched LUBIs is not sufficient to ensure the LU equivalence if the marginals of states are degenerate. We have constructed two-qutrit states in Example \ref{ex:c-ex}, two-qudit states in Example \ref{ex:cex-2}, and three-qubit states in Examle \ref{ex:cex-tri}, which are matched respectively but not LU equivalent.

There are some interesting problems remaining for future work. Since matched LUBIs are generally not sufficient to ensure the LU equivalence, it is naturally to ask the range of applicability of this criterion, i.e., characterizing the set of states where the LU equivalence can be determined by this criterion. From counterexamples proposed in Sec. \ref{sec:cex}, it would be interesting to find additional LU invariants to make up a complete set of LU invariants.

\acknowledgments

Y.S. is funded by the NNSF of China (Grant No. 12401597) and the Basic Research Program of Jiangsu (Grant No. BK20241603). L.C. is supported by the NNSF of China (Grant No. 12471427).

\appendix

\section{Verification of Example \ref{ex:c-ex} based on the known criterion}
\label{sec:verify2}

Here we show that the two states $\rho_{AB}$ and $\sigma_{AB}$ proposed in Example \ref{ex:c-ex} are not LU equivalent by the known criterion given by \cite[Eq. (7)]{bilu2012}. It suffices to check the invariants defined by the second line in \cite[Eq. (7)]{bilu2012}.

According to the two states in \eqref{eq:cex-2}, we obtain the following matrices:
\beq
\label{eq:ve-1}
\bal
A_1&=\frac{1}{\sqrt{3}}\diag(1,1,1), \quad A'_1=\frac{1}{\sqrt{3}}\diag(1,1,1), \\
A_2&=\frac{1}{\sqrt{3}}\diag(1,\omega,\omega^2), \quad 
A'_2=\frac{1}{\sqrt{3}}
\bma
0 & 1 & 0 \\
0 & 0 & 1 \\
1 & 0 & 0
\ema, \\
A_3&=\frac{1}{\sqrt{3}}\diag(1,\omega^2,\omega), \quad A'_3=\frac{1}{\sqrt{3}}
\bma
0 & 1 & 0 \\
0 & 0 & \omega \\
\omega^2 & 0 & 0
\ema.
\eal
\eeq
For each $1\leq j\leq 3$ we have $A_jA_j^\dg=\frac{1}{3}I_3$ and $A'_j(A'_j)^\dg=\frac{1}{3}I_3$. For each pair $(j,k)$ with $j\neq k$, we have
\beq
\label{eq:ve-2}
\bal
A_1A_2^\dg&=\frac{1}{\sqrt{3}}A_2^\dg,~~A_2A_1^\dg=\frac{1}{\sqrt{3}}A_2,  \\
A'_1(A'_2)^\dg&=\frac{1}{\sqrt{3}}(A'_2)^\dg,~~A'_2(A'_1)^\dg=\frac{1}{\sqrt{3}}A'_2,  \\
A_1A_3^\dg&=\frac{1}{\sqrt{3}}A_3^\dg, ~~A_3A_1^\dg=\frac{1}{\sqrt{3}}A_3, \\
A'_1(A'_3)^\dg&=\frac{1}{\sqrt{3}}(A'_3)^\dg,~~A'_3(A'_1)^\dg=\frac{1}{\sqrt{3}}A'_3,  \\
A_2A_3^\dg&=\frac{1}{\sqrt{3}}A_3, ~~A_3A_2^\dg=\frac{1}{\sqrt{3}}A_2, \\
A'_2(A'_3)^\dg&=\frac{1}{\sqrt{3}}A_3,~~A'_3(A'_2)^\dg=\frac{1}{\sqrt{3}}A_2.  \\
\eal
\eeq
By the above equation we find
\beq
\label{eq:ve-3}
\bal
&\tr[(A_1A_2^\dg)(A_3A_2^\dg)]=\frac{1}{3}\tr[A_2^\dg A_2]=\frac{1}{3}\tr[\frac{1}{3}I_3]=\frac{1}{3}, \\
&\tr\Big[\big(A'_1(A'_2)^\dg\big)\big(A'_3(A'_2)^\dg\big)\Big]=\frac{1}{3}\tr[(A'_2)^\dg A_2]=0.
\eal
\eeq
The two traces calculated by Eq. \eqref{eq:ve-3} contradcits one of the invariants defined by the second line in \cite[Eq. (7)]{bilu2012}. Thus, according to the criterion given by \cite{bilu2012} we also determine that the two states in \eqref{eq:cex-2} are not LU equivalent.


\bibliography{witness}

\end{document}

%% file: notation.tex
\usepackage[utf8]{inputenc} 
\usepackage{amsmath,amsfonts,amssymb,mathrsfs,amsthm,cases,tabularx} 
\usepackage{bm, bbm, dsfont} 
\usepackage{physics} 
\usepackage{graphicx} 
\usepackage{extarrows}
\usepackage{diagbox,threeparttable,dashrule,booktabs,dcolumn} 
\usepackage{xcolor} 
\usepackage{url} 
\usepackage{hyperref} 
\hypersetup{colorlinks,linkcolor={blue},citecolor={blue},urlcolor={red}}
\usepackage{cleveref} 

\crefname{equation}{Eq.}{Eqs.}
\crefname{section}{Sec.}{Secs.}
\crefname{definition}{Definition}{Definitions}
\crefname{proposition}{Proposition}{Propositions}
\crefname{lemma}{Lemma}{Lemmas}
\crefname{theorem}{Theorem}{Theorems}
\crefname{corollary}{Corollary}{Corollaries}
\crefname{conjecture}{Conjecture}{Conjectures}
\crefname{claim}{}{Claims}
\crefname{example}{Example}{Examples}

\newtheorem{definition}{Definition}
\newtheorem{proposition}[definition]{Proposition}
\newtheorem{lemma}[definition]{Lemma}

\newtheorem{theorem}[definition]{Theorem}
\newtheorem{corollary}[definition]{Corollary}

\newtheorem{example}[definition]{Example}

\newcommand{\nc}{\newcommand}

\def\red{\textcolor{red}}

\def\bpf{\begin{proof}}
\def\epf{\end{proof}}
\def\bea{\begin{eqnarray}}
\def\eea{\end{eqnarray}}
\def\beq{\begin{equation}}
\def\eeq{\end{equation}}
\def\bal{\begin{aligned}}
\def\eal{\end{aligned}}
\def\bma{\begin{pmatrix}}
\def\ema{\end{pmatrix}}

\def\bigox{\bigotimes}
\def\dg{\dagger}

\def\ox{\otimes}

\def\I{\mathds{I}}
\def\diag{\mathop{{\rm diag}}}

\def\a{\alpha}
\def\b{\beta}

\nc{\bbA}{{\mathbb A}}  \nc{\bbB}{{\mathbb B}}  \nc{\bbC}{{\mathbb C}}
\nc{\bbD}{{\mathbb D}}  \nc{\bbE}{{\mathbb E}}  \nc{\bbF}{{\mathbb F}}
\nc{\bbG}{{\mathbb G}}  \nc{\bbH}{{\mathbb H}}  \nc{\bbI}{{\mathbb I}}
\nc{\bbJ}{{\mathbb J}}  \nc{\bbK}{{\mathbb K}}  \nc{\bbL}{{\mathbb L}}
\nc{\bbM}{{\mathbb M}}  \nc{\bbN}{{\mathbb N}}  \nc{\bbO}{{\mathbb O}}
\nc{\bbP}{{\mathbb P}}  \nc{\bbQ}{{\mathbb Q}}  \nc{\bbR}{{\mathbb R}}
\nc{\bbS}{{\mathbb S}}  \nc{\bbT}{{\mathbb T}}  \nc{\bbU}{{\mathbb U}}
\nc{\bbV}{{\mathbb V}}  \nc{\bbW}{{\mathbb W}}  \nc{\bbX}{{\mathbb X}}
\nc{\bbY}{{\mathbb Y}}  \nc{\bbZ}{{\mathbb Z}}  

\nc{\bA}{{\bf A}}  \nc{\bB}{{\bf B}}  \nc{\bC}{{\bf C}}
\nc{\bD}{{\bf D}}  \nc{\bE}{{\bf E}}  \nc{\bF}{{\bf F}}
\nc{\bG}{{\bf G}}  \nc{\bH}{{\bf H}}  \nc{\bI}{{\bf I}}
\nc{\bJ}{{\bf J}}  \nc{\bK}{{\bf K}}  \nc{\bL}{{\bf L}}
\nc{\bM}{{\bf M}}  \nc{\bN}{{\bf N}}  \nc{\bO}{{\bf O}}
\nc{\bP}{{\bf P}}  \nc{\bQ}{{\bf Q}}  \nc{\bR}{{\bf R}}
\nc{\bS}{{\bf S}}  \nc{\bT}{{\bf T}}  \nc{\bU}{{\bf U}}
\nc{\bV}{{\bf V}}  \nc{\bW}{{\bf W}}  \nc{\bX}{{\bf X}}
\nc{\bY}{{\bf Y}}  \nc{\bZ}{{\bf Z}}  

\nc{\cA}{{\cal A}}  \nc{\cB}{{\cal B}}  \nc{\cC}{{\cal C}}
\nc{\cD}{{\cal D}}  \nc{\cE}{{\cal E}}  \nc{\cF}{{\cal F}}
\nc{\cG}{{\cal G}}  \nc{\cH}{{\cal H}}  \nc{\cI}{{\cal I}}
\nc{\cJ}{{\cal J}}  \nc{\cK}{{\cal K}}  \nc{\cL}{{\cal L}}
\nc{\cM}{{\cal M}}  \nc{\cN}{{\cal N}}  \nc{\cO}{{\cal O}}
\nc{\cP}{{\cal P}}  \nc{\cQ}{{\cal Q}}  \nc{\cR}{{\cal R}}
\nc{\cS}{{\cal S}}  \nc{\cT}{{\cal T}}  \nc{\cU}{{\cal U}}
\nc{\cV}{{\cal V}}  \nc{\cW}{{\cal W}}  \nc{\cX}{{\cal X}}
\nc{\cY}{{\cal Y}}  \nc{\cZ}{{\cal Z}}  

\def\ox{\otimes}
\def\bigox{\bigotimes}
\def\dg{\dagger}




%% file: LU_criterion_26v5.bbl
\begin{thebibliography}{36}%
\makeatletter
\providecommand \@ifxundefined [1]{%
 \@ifx{#1\undefined}
}%
\providecommand \@ifnum [1]{%
 \ifnum #1\expandafter \@firstoftwo
 \else \expandafter \@secondoftwo
 \fi
}%
\providecommand \@ifx [1]{%
 \ifx #1\expandafter \@firstoftwo
 \else \expandafter \@secondoftwo
 \fi
}%
\providecommand \natexlab [1]{#1}%
\providecommand \enquote  [1]{``#1''}%
\providecommand \bibnamefont  [1]{#1}%
\providecommand \bibfnamefont [1]{#1}%
\providecommand \citenamefont [1]{#1}%
\providecommand \href@noop [0]{\@secondoftwo}%
\providecommand \href [0]{\begingroup \@sanitize@url \@href}%
\providecommand \@href[1]{\@@startlink{#1}\@@href}%
\providecommand \@@href[1]{\endgroup#1\@@endlink}%
\providecommand \@sanitize@url [0]{\catcode `\\12\catcode `\$12\catcode `\&12\catcode `\#12\catcode `\^12\catcode `\_12\catcode `\%12\relax}%
\providecommand \@@startlink[1]{}%
\providecommand \@@endlink[0]{}%
\providecommand \url  [0]{\begingroup\@sanitize@url \@url }%
\providecommand \@url [1]{\endgroup\@href {#1}{\urlprefix }}%
\providecommand \urlprefix  [0]{URL }%
\providecommand \Eprint [0]{\href }%
\providecommand \doibase [0]{https://doi.org/}%
\providecommand \selectlanguage [0]{\@gobble}%
\providecommand \bibinfo  [0]{\@secondoftwo}%
\providecommand \bibfield  [0]{\@secondoftwo}%
\providecommand \translation [1]{[#1]}%
\providecommand \BibitemOpen [0]{}%
\providecommand \bibitemStop [0]{}%
\providecommand \bibitemNoStop [0]{.\EOS\space}%
\providecommand \EOS [0]{\spacefactor3000\relax}%
\providecommand \BibitemShut  [1]{\csname bibitem#1\endcsname}%
\let\auto@bib@innerbib\@empty
\bibitem [{\citenamefont {Hu}\ \emph {et~al.}(2023)\citenamefont {Hu}, \citenamefont {Guo}, \citenamefont {Liu}, \citenamefont {Li},\ and\ \citenamefont {Guo}}]{telepor-naturereview2023}%
  \BibitemOpen
  \bibfield  {author} {\bibinfo {author} {\bibfnamefont {X.-M.}\ \bibnamefont {Hu}}, \bibinfo {author} {\bibfnamefont {Y.}~\bibnamefont {Guo}}, \bibinfo {author} {\bibfnamefont {B.-H.}\ \bibnamefont {Liu}}, \bibinfo {author} {\bibfnamefont {C.-F.}\ \bibnamefont {Li}},\ and\ \bibinfo {author} {\bibfnamefont {G.-C.}\ \bibnamefont {Guo}},\ }\bibfield  {title} {\bibinfo {title} {Progress in quantum teleportation},\ }\href {https://doi.org/10.1038/s42254-023-00588-x} {\bibfield  {journal} {\bibinfo  {journal} {Nature Reviews Physics}\ }\textbf {\bibinfo {volume} {5}},\ \bibinfo {pages} {339} (\bibinfo {year} {2023})}\BibitemShut {NoStop}%
\bibitem [{\citenamefont {Hayashi}\ and\ \citenamefont {Wang}(2022)}]{densecoding2022}%
  \BibitemOpen
  \bibfield  {author} {\bibinfo {author} {\bibfnamefont {M.}~\bibnamefont {Hayashi}}\ and\ \bibinfo {author} {\bibfnamefont {K.}~\bibnamefont {Wang}},\ }\bibfield  {title} {\bibinfo {title} {Dense coding with locality restriction on decoders: Quantum encoders versus superquantum encoders},\ }\href {https://doi.org/10.1103/PRXQuantum.3.030346} {\bibfield  {journal} {\bibinfo  {journal} {PRX Quantum}\ }\textbf {\bibinfo {volume} {3}},\ \bibinfo {pages} {030346} (\bibinfo {year} {2022})}\BibitemShut {NoStop}%
\bibitem [{\citenamefont {Yin}\ \emph {et~al.}(2020)\citenamefont {Yin}, \citenamefont {Li}, \citenamefont {Liao}, \citenamefont {Yang}, \citenamefont {Cao}, \citenamefont {Zhang}, \citenamefont {Ren}, \citenamefont {Cai}, \citenamefont {Liu}, \citenamefont {Li}, \citenamefont {Shu}, \citenamefont {Huang}, \citenamefont {Deng}, \citenamefont {Li}, \citenamefont {Zhang}, \citenamefont {Liu}, \citenamefont {Chen}, \citenamefont {Lu}, \citenamefont {Wang}, \citenamefont {Xu}, \citenamefont {Wang}, \citenamefont {Peng}, \citenamefont {Ekert},\ and\ \citenamefont {Pan}}]{qcrypto2020}%
  \BibitemOpen
  \bibfield  {author} {\bibinfo {author} {\bibfnamefont {J.}~\bibnamefont {Yin}}, \bibinfo {author} {\bibfnamefont {Y.-H.}\ \bibnamefont {Li}}, \bibinfo {author} {\bibfnamefont {S.-K.}\ \bibnamefont {Liao}}, \bibinfo {author} {\bibfnamefont {M.}~\bibnamefont {Yang}}, \bibinfo {author} {\bibfnamefont {Y.}~\bibnamefont {Cao}}, \bibinfo {author} {\bibfnamefont {L.}~\bibnamefont {Zhang}}, \bibinfo {author} {\bibfnamefont {J.-G.}\ \bibnamefont {Ren}}, \bibinfo {author} {\bibfnamefont {W.-Q.}\ \bibnamefont {Cai}}, \bibinfo {author} {\bibfnamefont {W.-Y.}\ \bibnamefont {Liu}}, \bibinfo {author} {\bibfnamefont {S.-L.}\ \bibnamefont {Li}}, \bibinfo {author} {\bibfnamefont {R.}~\bibnamefont {Shu}}, \bibinfo {author} {\bibfnamefont {Y.-M.}\ \bibnamefont {Huang}}, \bibinfo {author} {\bibfnamefont {L.}~\bibnamefont {Deng}}, \bibinfo {author} {\bibfnamefont {L.}~\bibnamefont {Li}}, \bibinfo {author} {\bibfnamefont {Q.}~\bibnamefont {Zhang}}, \bibinfo {author} {\bibfnamefont {N.-L.}\ \bibnamefont {Liu}}, \bibinfo {author}
  {\bibfnamefont {Y.-A.}\ \bibnamefont {Chen}}, \bibinfo {author} {\bibfnamefont {C.-Y.}\ \bibnamefont {Lu}}, \bibinfo {author} {\bibfnamefont {X.-B.}\ \bibnamefont {Wang}}, \bibinfo {author} {\bibfnamefont {F.}~\bibnamefont {Xu}}, \bibinfo {author} {\bibfnamefont {J.-Y.}\ \bibnamefont {Wang}}, \bibinfo {author} {\bibfnamefont {C.-Z.}\ \bibnamefont {Peng}}, \bibinfo {author} {\bibfnamefont {A.~K.}\ \bibnamefont {Ekert}},\ and\ \bibinfo {author} {\bibfnamefont {J.-W.}\ \bibnamefont {Pan}},\ }\bibfield  {title} {\bibinfo {title} {Entanglement-based secure quantum cryptography over 1,120 kilometres},\ }\href {https://doi.org/10.1038/s41586-020-2401-y} {\bibfield  {journal} {\bibinfo  {journal} {Nature}\ }\textbf {\bibinfo {volume} {582}},\ \bibinfo {pages} {501} (\bibinfo {year} {2020})}\BibitemShut {NoStop}%
\bibitem [{\citenamefont {Acin}\ \emph {et~al.}(2000)\citenamefont {Acin}, \citenamefont {Andrianov}, \citenamefont {Costa}, \citenamefont {Jane}, \citenamefont {Latorre},\ and\ \citenamefont {Tarrach}}]{3qubitlu2000}%
  \BibitemOpen
  \bibfield  {author} {\bibinfo {author} {\bibfnamefont {A.}~\bibnamefont {Acin}}, \bibinfo {author} {\bibfnamefont {A.}~\bibnamefont {Andrianov}}, \bibinfo {author} {\bibfnamefont {L.}~\bibnamefont {Costa}}, \bibinfo {author} {\bibfnamefont {E.}~\bibnamefont {Jane}}, \bibinfo {author} {\bibfnamefont {J.~I.}\ \bibnamefont {Latorre}},\ and\ \bibinfo {author} {\bibfnamefont {R.}~\bibnamefont {Tarrach}},\ }\bibfield  {title} {\bibinfo {title} {Generalized schmidt decomposition and classification of three-quantum-bit states},\ }\href {https://doi.org/10.1103/PhysRevLett.85.1560} {\bibfield  {journal} {\bibinfo  {journal} {Phys. Rev. Lett.}\ }\textbf {\bibinfo {volume} {85}},\ \bibinfo {pages} {1560} (\bibinfo {year} {2000})}\BibitemShut {NoStop}%
\bibitem [{\citenamefont {Acin}\ \emph {et~al.}(2001)\citenamefont {Acin}, \citenamefont {Bruss}, \citenamefont {Lewenstein},\ and\ \citenamefont {Sanpera}}]{3qubitinequiv2001}%
  \BibitemOpen
  \bibfield  {author} {\bibinfo {author} {\bibfnamefont {A.}~\bibnamefont {Acin}}, \bibinfo {author} {\bibfnamefont {D.}~\bibnamefont {Bruss}}, \bibinfo {author} {\bibfnamefont {M.}~\bibnamefont {Lewenstein}},\ and\ \bibinfo {author} {\bibfnamefont {A.}~\bibnamefont {Sanpera}},\ }\bibfield  {title} {\bibinfo {title} {Classification of mixed three-qubit states},\ }\href {https://doi.org/10.1103/PhysRevLett.87.040401} {\bibfield  {journal} {\bibinfo  {journal} {Phys. Rev. Lett.}\ }\textbf {\bibinfo {volume} {87}},\ \bibinfo {pages} {040401} (\bibinfo {year} {2001})}\BibitemShut {NoStop}%
\bibitem [{\citenamefont {Kraus}(2010{\natexlab{a}})}]{luequiv2010}%
  \BibitemOpen
  \bibfield  {author} {\bibinfo {author} {\bibfnamefont {B.}~\bibnamefont {Kraus}},\ }\bibfield  {title} {\bibinfo {title} {Local unitary equivalence of multipartite pure states},\ }\href {https://doi.org/10.1103/PhysRevLett.104.020504} {\bibfield  {journal} {\bibinfo  {journal} {Phys. Rev. Lett.}\ }\textbf {\bibinfo {volume} {104}},\ \bibinfo {pages} {020504} (\bibinfo {year} {2010}{\natexlab{a}})}\BibitemShut {NoStop}%
\bibitem [{\citenamefont {Kraus}(2010{\natexlab{b}})}]{pslu-2010}%
  \BibitemOpen
  \bibfield  {author} {\bibinfo {author} {\bibfnamefont {B.}~\bibnamefont {Kraus}},\ }\bibfield  {title} {\bibinfo {title} {Local unitary equivalence and entanglement of multipartite pure states},\ }\href {https://doi.org/10.1103/PhysRevA.82.032121} {\bibfield  {journal} {\bibinfo  {journal} {Phys. Rev. A}\ }\textbf {\bibinfo {volume} {82}},\ \bibinfo {pages} {032121} (\bibinfo {year} {2010}{\natexlab{b}})}\BibitemShut {NoStop}%
\bibitem [{\citenamefont {Liu}\ \emph {et~al.}(2012)\citenamefont {Liu}, \citenamefont {Li}, \citenamefont {Li},\ and\ \citenamefont {Qiao}}]{mpsinequivlu2012}%
  \BibitemOpen
  \bibfield  {author} {\bibinfo {author} {\bibfnamefont {B.}~\bibnamefont {Liu}}, \bibinfo {author} {\bibfnamefont {J.-L.}\ \bibnamefont {Li}}, \bibinfo {author} {\bibfnamefont {X.}~\bibnamefont {Li}},\ and\ \bibinfo {author} {\bibfnamefont {C.-F.}\ \bibnamefont {Qiao}},\ }\bibfield  {title} {\bibinfo {title} {Local unitary classification of arbitrary dimensional multipartite pure states},\ }\href {https://doi.org/10.1103/PhysRevLett.108.050501} {\bibfield  {journal} {\bibinfo  {journal} {Phys. Rev. Lett.}\ }\textbf {\bibinfo {volume} {108}},\ \bibinfo {pages} {050501} (\bibinfo {year} {2012})}\BibitemShut {NoStop}%
\bibitem [{\citenamefont {Li}(2018)}]{luspinflip2018}%
  \BibitemOpen
  \bibfield  {author} {\bibinfo {author} {\bibfnamefont {D.}~\bibnamefont {Li}},\ }\bibfield  {title} {\bibinfo {title} {Stochastic local operations and classical communication (slocc) and local unitary operations (lu) classifications of n qubits via ranks and singular values of the spin-flipping matrices},\ }\href {https://doi.org/10.1007/s11128-018-1900-3} {\bibfield  {journal} {\bibinfo  {journal} {Quantum Information Processing}\ }\textbf {\bibinfo {volume} {17}},\ \bibinfo {pages} {132} (\bibinfo {year} {2018})}\BibitemShut {NoStop}%
\bibitem [{\citenamefont {Makhlin}(2002)}]{Makhlin2002}%
  \BibitemOpen
  \bibfield  {author} {\bibinfo {author} {\bibfnamefont {Y.}~\bibnamefont {Makhlin}},\ }\bibfield  {title} {\bibinfo {title} {Nonlocal properties of two-qubit gates and mixed states, and the optimization of quantum computations},\ }\href {https://doi.org/10.1023/A:1022144002391} {\bibfield  {journal} {\bibinfo  {journal} {Quantum Information Processing}\ }\textbf {\bibinfo {volume} {1}},\ \bibinfo {pages} {243} (\bibinfo {year} {2002})}\BibitemShut {NoStop}%
\bibitem [{\citenamefont {Sun}\ \emph {et~al.}(2017)\citenamefont {Sun}, \citenamefont {Fei},\ and\ \citenamefont {Wang}}]{3qbitinv2017}%
  \BibitemOpen
  \bibfield  {author} {\bibinfo {author} {\bibfnamefont {B.-Z.}\ \bibnamefont {Sun}}, \bibinfo {author} {\bibfnamefont {S.-M.}\ \bibnamefont {Fei}},\ and\ \bibinfo {author} {\bibfnamefont {Z.-X.}\ \bibnamefont {Wang}},\ }\bibfield  {title} {\bibinfo {title} {On local unitary equivalence of two and three-qubit states},\ }\href {https://doi.org/10.1038/s41598-017-04717-2} {\bibfield  {journal} {\bibinfo  {journal} {Scientific Reports}\ }\textbf {\bibinfo {volume} {7}},\ \bibinfo {pages} {4869} (\bibinfo {year} {2017})}\BibitemShut {NoStop}%
\bibitem [{\citenamefont {Zhou}\ \emph {et~al.}(2012)\citenamefont {Zhou}, \citenamefont {Zhang}, \citenamefont {Fei}, \citenamefont {Jing},\ and\ \citenamefont {Li-Jost}}]{bilu2012}%
  \BibitemOpen
  \bibfield  {author} {\bibinfo {author} {\bibfnamefont {C.}~\bibnamefont {Zhou}}, \bibinfo {author} {\bibfnamefont {T.-G.}\ \bibnamefont {Zhang}}, \bibinfo {author} {\bibfnamefont {S.-M.}\ \bibnamefont {Fei}}, \bibinfo {author} {\bibfnamefont {N.}~\bibnamefont {Jing}},\ and\ \bibinfo {author} {\bibfnamefont {X.}~\bibnamefont {Li-Jost}},\ }\bibfield  {title} {\bibinfo {title} {Local unitary equivalence of arbitrary dimensional bipartite quantum states},\ }\href {https://doi.org/10.1103/PhysRevA.86.010303} {\bibfield  {journal} {\bibinfo  {journal} {Phys. Rev. A}\ }\textbf {\bibinfo {volume} {86}},\ \bibinfo {pages} {010303} (\bibinfo {year} {2012})}\BibitemShut {NoStop}%
\bibitem [{\citenamefont {Zhang}\ \emph {et~al.}(2013)\citenamefont {Zhang}, \citenamefont {Zhao}, \citenamefont {Li}, \citenamefont {Fei},\ and\ \citenamefont {Li-Jost}}]{msinv2013}%
  \BibitemOpen
  \bibfield  {author} {\bibinfo {author} {\bibfnamefont {T.-G.}\ \bibnamefont {Zhang}}, \bibinfo {author} {\bibfnamefont {M.-J.}\ \bibnamefont {Zhao}}, \bibinfo {author} {\bibfnamefont {M.}~\bibnamefont {Li}}, \bibinfo {author} {\bibfnamefont {S.-M.}\ \bibnamefont {Fei}},\ and\ \bibinfo {author} {\bibfnamefont {X.}~\bibnamefont {Li-Jost}},\ }\bibfield  {title} {\bibinfo {title} {Criterion of local unitary equivalence for multipartite states},\ }\href {https://doi.org/10.1103/PhysRevA.88.042304} {\bibfield  {journal} {\bibinfo  {journal} {Phys. Rev. A}\ }\textbf {\bibinfo {volume} {88}},\ \bibinfo {pages} {042304} (\bibinfo {year} {2013})}\BibitemShut {NoStop}%
\bibitem [{\citenamefont {Bhosale}\ \emph {et~al.}(2013)\citenamefont {Bhosale}, \citenamefont {Shuddhodan},\ and\ \citenamefont {Lakshminarayan}}]{PTLU2013}%
  \BibitemOpen
  \bibfield  {author} {\bibinfo {author} {\bibfnamefont {U.~T.}\ \bibnamefont {Bhosale}}, \bibinfo {author} {\bibfnamefont {K.~V.}\ \bibnamefont {Shuddhodan}},\ and\ \bibinfo {author} {\bibfnamefont {A.}~\bibnamefont {Lakshminarayan}},\ }\bibfield  {title} {\bibinfo {title} {Using partial transpose and realignment to generate local unitary invariants},\ }\href {https://doi.org/10.1103/PhysRevA.87.052311} {\bibfield  {journal} {\bibinfo  {journal} {Phys. Rev. A}\ }\textbf {\bibinfo {volume} {87}},\ \bibinfo {pages} {052311} (\bibinfo {year} {2013})}\BibitemShut {NoStop}%
\bibitem [{\citenamefont {Jing}\ \emph {et~al.}(2015)\citenamefont {Jing}, \citenamefont {Fei}, \citenamefont {Li}, \citenamefont {Li-Jost},\ and\ \citenamefont {Zhang}}]{mqbitinv2015}%
  \BibitemOpen
  \bibfield  {author} {\bibinfo {author} {\bibfnamefont {N.}~\bibnamefont {Jing}}, \bibinfo {author} {\bibfnamefont {S.-M.}\ \bibnamefont {Fei}}, \bibinfo {author} {\bibfnamefont {M.}~\bibnamefont {Li}}, \bibinfo {author} {\bibfnamefont {X.}~\bibnamefont {Li-Jost}},\ and\ \bibinfo {author} {\bibfnamefont {T.}~\bibnamefont {Zhang}},\ }\bibfield  {title} {\bibinfo {title} {Local unitary invariants of generic multiqubit states},\ }\href {https://doi.org/10.1103/PhysRevA.92.022306} {\bibfield  {journal} {\bibinfo  {journal} {Phys. Rev. A}\ }\textbf {\bibinfo {volume} {92}},\ \bibinfo {pages} {022306} (\bibinfo {year} {2015})}\BibitemShut {NoStop}%
\bibitem [{\citenamefont {Dobes}\ and\ \citenamefont {Jing}(2025)}]{luhypermatrix2025}%
  \BibitemOpen
  \bibfield  {author} {\bibinfo {author} {\bibfnamefont {I.}~\bibnamefont {Dobes}}\ and\ \bibinfo {author} {\bibfnamefont {N.}~\bibnamefont {Jing}},\ }\bibfield  {title} {\bibinfo {title} {Local unitary equivalence of tripartite quantum states in terms of trace identities},\ }\href {https://doi.org/10.1007/s11128-025-04784-9} {\bibfield  {journal} {\bibinfo  {journal} {Quantum Information Processing}\ }\textbf {\bibinfo {volume} {24}},\ \bibinfo {pages} {158} (\bibinfo {year} {2025})}\BibitemShut {NoStop}%
\bibitem [{\citenamefont {Chu}\ \emph {et~al.}(2026)\citenamefont {Chu}, \citenamefont {Cui}, \citenamefont {Xie}, \citenamefont {Liu},\ and\ \citenamefont {Fei}}]{LUdegen2026}%
  \BibitemOpen
  \bibfield  {author} {\bibinfo {author} {\bibfnamefont {Y.}~\bibnamefont {Chu}}, \bibinfo {author} {\bibfnamefont {C.}~\bibnamefont {Cui}}, \bibinfo {author} {\bibfnamefont {Y.}~\bibnamefont {Xie}}, \bibinfo {author} {\bibfnamefont {M.}~\bibnamefont {Liu}},\ and\ \bibinfo {author} {\bibfnamefont {S.-M.}\ \bibnamefont {Fei}},\ }\bibfield  {title} {\bibinfo {title} {Identifying local unitary equivalence based on reduction of quantum states},\ }\href {https://doi.org/10.1103/kjgk-cpw7} {\bibfield  {journal} {\bibinfo  {journal} {Phys. Rev. A}\ }\textbf {\bibinfo {volume} {114}},\ \bibinfo {pages} {012451} (\bibinfo {year} {2026})}\BibitemShut {NoStop}%
\bibitem [{\citenamefont {Wang}\ \emph {et~al.}(2025)\citenamefont {Wang}, \citenamefont {Yuan}, \citenamefont {Li}, \citenamefont {Yang},\ and\ \citenamefont {Fei}}]{LUGBS2025}%
  \BibitemOpen
  \bibfield  {author} {\bibinfo {author} {\bibfnamefont {C.-H.}\ \bibnamefont {Wang}}, \bibinfo {author} {\bibfnamefont {J.-T.}\ \bibnamefont {Yuan}}, \bibinfo {author} {\bibfnamefont {M.-S.}\ \bibnamefont {Li}}, \bibinfo {author} {\bibfnamefont {Y.-H.}\ \bibnamefont {Yang}},\ and\ \bibinfo {author} {\bibfnamefont {S.-M.}\ \bibnamefont {Fei}},\ }\bibfield  {title} {\bibinfo {title} {Local unitary classification of sets of generalized bell states in $\mathbb{C}^{d}\otimes \mathbb{C}^{d}$},\ }\href@noop {} {\bibfield  {journal} {\bibinfo  {journal} {EPJ Quantum Technology}\ }\textbf {\bibinfo {volume} {12}},\ \bibinfo {pages} {87} (\bibinfo {year} {2025})}\BibitemShut {NoStop}%
\bibitem [{\citenamefont {Shen}\ and\ \citenamefont {Chen}(2025)}]{yiprr2025}%
  \BibitemOpen
  \bibfield  {author} {\bibinfo {author} {\bibfnamefont {Y.}~\bibnamefont {Shen}}\ and\ \bibinfo {author} {\bibfnamefont {L.}~\bibnamefont {Chen}},\ }\bibfield  {title} {\bibinfo {title} {Decidabilities of local unitary equivalence for entanglement witnesses and states},\ }\href {https://doi.org/10.1103/bzfp-wxbt} {\bibfield  {journal} {\bibinfo  {journal} {Phys. Rev. Res.}\ }\textbf {\bibinfo {volume} {7}},\ \bibinfo {pages} {033258} (\bibinfo {year} {2025})}\BibitemShut {NoStop}%
\bibitem [{\citenamefont {Buhrman}\ \emph {et~al.}(2001)\citenamefont {Buhrman}, \citenamefont {Cleve}, \citenamefont {Watrous},\ and\ \citenamefont {de~Wolf}}]{qfinprint2001}%
  \BibitemOpen
  \bibfield  {author} {\bibinfo {author} {\bibfnamefont {H.}~\bibnamefont {Buhrman}}, \bibinfo {author} {\bibfnamefont {R.}~\bibnamefont {Cleve}}, \bibinfo {author} {\bibfnamefont {J.}~\bibnamefont {Watrous}},\ and\ \bibinfo {author} {\bibfnamefont {R.}~\bibnamefont {de~Wolf}},\ }\bibfield  {title} {\bibinfo {title} {Quantum fingerprinting},\ }\href {https://doi.org/10.1103/PhysRevLett.87.167902} {\bibfield  {journal} {\bibinfo  {journal} {Phys. Rev. Lett.}\ }\textbf {\bibinfo {volume} {87}},\ \bibinfo {pages} {167902} (\bibinfo {year} {2001})}\BibitemShut {NoStop}%
\bibitem [{\citenamefont {Wagner}\ \emph {et~al.}(2024)\citenamefont {Wagner}, \citenamefont {Schwartzman-Nowik}, \citenamefont {Paiva}, \citenamefont {Te’eni}, \citenamefont {Ruiz-Molero}, \citenamefont {Barbosa}, \citenamefont {Cohen},\ and\ \citenamefont {Galvão}}]{KDBI2024}%
  \BibitemOpen
  \bibfield  {author} {\bibinfo {author} {\bibfnamefont {R.}~\bibnamefont {Wagner}}, \bibinfo {author} {\bibfnamefont {Z.}~\bibnamefont {Schwartzman-Nowik}}, \bibinfo {author} {\bibfnamefont {I.~L.}\ \bibnamefont {Paiva}}, \bibinfo {author} {\bibfnamefont {A.}~\bibnamefont {Te’eni}}, \bibinfo {author} {\bibfnamefont {A.}~\bibnamefont {Ruiz-Molero}}, \bibinfo {author} {\bibfnamefont {R.~S.}\ \bibnamefont {Barbosa}}, \bibinfo {author} {\bibfnamefont {E.}~\bibnamefont {Cohen}},\ and\ \bibinfo {author} {\bibfnamefont {E.~F.}\ \bibnamefont {Galvão}},\ }\bibfield  {title} {\bibinfo {title} {Quantum circuits for measuring weak values, kirkwood–dirac quasiprobability distributions, and state spectra},\ }\href {https://doi.org/10.1088/2058-9565/ad124c} {\bibfield  {journal} {\bibinfo  {journal} {Quantum Science and Technology}\ }\textbf {\bibinfo {volume} {9}},\ \bibinfo {pages} {015030} (\bibinfo {year} {2024})}\BibitemShut {NoStop}%
\bibitem [{\citenamefont {Fernandes}\ \emph {et~al.}(2024)\citenamefont {Fernandes}, \citenamefont {Wagner}, \citenamefont {Novo},\ and\ \citenamefont {Galv\~ao}}]{BIIWprl2024}%
  \BibitemOpen
  \bibfield  {author} {\bibinfo {author} {\bibfnamefont {C.}~\bibnamefont {Fernandes}}, \bibinfo {author} {\bibfnamefont {R.}~\bibnamefont {Wagner}}, \bibinfo {author} {\bibfnamefont {L.}~\bibnamefont {Novo}},\ and\ \bibinfo {author} {\bibfnamefont {E.~F.}\ \bibnamefont {Galv\~ao}},\ }\bibfield  {title} {\bibinfo {title} {Unitary-invariant witnesses of quantum imaginarity},\ }\href {https://doi.org/10.1103/PhysRevLett.133.190201} {\bibfield  {journal} {\bibinfo  {journal} {Phys. Rev. Lett.}\ }\textbf {\bibinfo {volume} {133}},\ \bibinfo {pages} {190201} (\bibinfo {year} {2024})}\BibitemShut {NoStop}%
\bibitem [{\citenamefont {Li}\ and\ \citenamefont {Tan}(2025)}]{BIIWpra2025}%
  \BibitemOpen
  \bibfield  {author} {\bibinfo {author} {\bibfnamefont {M.-S.}\ \bibnamefont {Li}}\ and\ \bibinfo {author} {\bibfnamefont {Y.-X.}\ \bibnamefont {Tan}},\ }\bibfield  {title} {\bibinfo {title} {Bargmann invariants for quantum imaginarity},\ }\href {https://doi.org/10.1103/PhysRevA.111.022409} {\bibfield  {journal} {\bibinfo  {journal} {Phys. Rev. A}\ }\textbf {\bibinfo {volume} {111}},\ \bibinfo {pages} {022409} (\bibinfo {year} {2025})}\BibitemShut {NoStop}%
\bibitem [{\citenamefont {Li}\ \emph {et~al.}(2026)\citenamefont {Li}, \citenamefont {Wagner},\ and\ \citenamefont {Zhang}}]{BIIWpra2026}%
  \BibitemOpen
  \bibfield  {author} {\bibinfo {author} {\bibfnamefont {M.-S.}\ \bibnamefont {Li}}, \bibinfo {author} {\bibfnamefont {R.}~\bibnamefont {Wagner}},\ and\ \bibinfo {author} {\bibfnamefont {L.}~\bibnamefont {Zhang}},\ }\bibfield  {title} {\bibinfo {title} {Multistate imaginarity and coherence in qubit systems},\ }\href {https://doi.org/10.1103/tpgw-v6ht} {\bibfield  {journal} {\bibinfo  {journal} {Phys. Rev. A}\ }\textbf {\bibinfo {volume} {113}},\ \bibinfo {pages} {012428} (\bibinfo {year} {2026})}\BibitemShut {NoStop}%
\bibitem [{\citenamefont {Guo}\ \emph {et~al.}(2026)\citenamefont {Guo}, \citenamefont {Pan}, \citenamefont {Wang},\ and\ \citenamefont {Du}}]{SetIguo2026}%
  \BibitemOpen
  \bibfield  {author} {\bibinfo {author} {\bibfnamefont {Y.}~\bibnamefont {Guo}}, \bibinfo {author} {\bibfnamefont {J.}~\bibnamefont {Pan}}, \bibinfo {author} {\bibfnamefont {Y.}~\bibnamefont {Wang}},\ and\ \bibinfo {author} {\bibfnamefont {S.}~\bibnamefont {Du}},\ }\bibfield  {title} {\bibinfo {title} {Measure of set imaginarity},\ }\href {https://doi.org/10.1103/1yc6-lzfk} {\bibfield  {journal} {\bibinfo  {journal} {Phys. Rev. A}\ }\textbf {\bibinfo {volume} {114}},\ \bibinfo {pages} {022453} (\bibinfo {year} {2026})}\BibitemShut {NoStop}%
\bibitem [{\citenamefont {Quek}\ \emph {et~al.}(2024)\citenamefont {Quek}, \citenamefont {Kaur},\ and\ \citenamefont {Wilde}}]{multivariatetr2024}%
  \BibitemOpen
  \bibfield  {author} {\bibinfo {author} {\bibfnamefont {Y.}~\bibnamefont {Quek}}, \bibinfo {author} {\bibfnamefont {E.}~\bibnamefont {Kaur}},\ and\ \bibinfo {author} {\bibfnamefont {M.~M.}\ \bibnamefont {Wilde}},\ }\bibfield  {title} {\bibinfo {title} {Multivariate trace estimation in constant quantum depth},\ }\href {https://doi.org/10.22331/q-2024-01-10-1220} {\bibfield  {journal} {\bibinfo  {journal} {{Quantum}}\ }\textbf {\bibinfo {volume} {8}},\ \bibinfo {pages} {1220} (\bibinfo {year} {2024})}\BibitemShut {NoStop}%
\bibitem [{\citenamefont {Zhang}\ \emph {et~al.}(2025{\natexlab{a}})\citenamefont {Zhang}, \citenamefont {Xie},\ and\ \citenamefont {Li}}]{BIchara202504}%
  \BibitemOpen
  \bibfield  {author} {\bibinfo {author} {\bibfnamefont {L.}~\bibnamefont {Zhang}}, \bibinfo {author} {\bibfnamefont {B.}~\bibnamefont {Xie}},\ and\ \bibinfo {author} {\bibfnamefont {B.}~\bibnamefont {Li}},\ }\bibfield  {title} {\bibinfo {title} {Geometry of sets of bargmann invariants},\ }\href {https://doi.org/10.1103/PhysRevA.111.042417} {\bibfield  {journal} {\bibinfo  {journal} {Phys. Rev. A}\ }\textbf {\bibinfo {volume} {111}},\ \bibinfo {pages} {042417} (\bibinfo {year} {2025}{\natexlab{a}})}\BibitemShut {NoStop}%
\bibitem [{\citenamefont {Pratapsi}\ \emph {et~al.}(2025)\citenamefont {Pratapsi}, \citenamefont {Gouveia}, \citenamefont {Novo},\ and\ \citenamefont {Galv\~ao}}]{BIchara202510}%
  \BibitemOpen
  \bibfield  {author} {\bibinfo {author} {\bibfnamefont {S.~S.}\ \bibnamefont {Pratapsi}}, \bibinfo {author} {\bibfnamefont {J.~a.}\ \bibnamefont {Gouveia}}, \bibinfo {author} {\bibfnamefont {L.}~\bibnamefont {Novo}},\ and\ \bibinfo {author} {\bibfnamefont {E.~F.}\ \bibnamefont {Galv\~ao}},\ }\bibfield  {title} {\bibinfo {title} {Elementary characterization of bargmann invariants},\ }\href {https://doi.org/10.1103/hsnv-wpt3} {\bibfield  {journal} {\bibinfo  {journal} {Phys. Rev. A}\ }\textbf {\bibinfo {volume} {112}},\ \bibinfo {pages} {042421} (\bibinfo {year} {2025})}\BibitemShut {NoStop}%
\bibitem [{\citenamefont {Xu}(2026{\natexlab{a}})}]{BIcharaXU202601}%
  \BibitemOpen
  \bibfield  {author} {\bibinfo {author} {\bibfnamefont {J.}~\bibnamefont {Xu}},\ }\bibfield  {title} {\bibinfo {title} {Numerical ranges of bargmann invariants},\ }\href {https://doi.org/https://doi.org/10.1016/j.physleta.2025.131091} {\bibfield  {journal} {\bibinfo  {journal} {Physics Letters A}\ }\textbf {\bibinfo {volume} {565}},\ \bibinfo {pages} {131091} (\bibinfo {year} {2026}{\natexlab{a}})}\BibitemShut {NoStop}%
\bibitem [{\citenamefont {Xu}(2026{\natexlab{b}})}]{BIcharaXu202602}%
  \BibitemOpen
  \bibfield  {author} {\bibinfo {author} {\bibfnamefont {J.}~\bibnamefont {Xu}},\ }\bibfield  {title} {\bibinfo {title} {Bargmann invariants of gaussian states},\ }\href {https://doi.org/10.1063/5.0306797} {\bibfield  {journal} {\bibinfo  {journal} {Journal of Mathematical Physics}\ }\textbf {\bibinfo {volume} {67}},\ \bibinfo {pages} {052101} (\bibinfo {year} {2026}{\natexlab{b}})}\BibitemShut {NoStop}%
\bibitem [{\citenamefont {Zhang}\ and\ \citenamefont {Xie}(2026)}]{BIZL2026review}%
  \BibitemOpen
  \bibfield  {author} {\bibinfo {author} {\bibfnamefont {L.}~\bibnamefont {Zhang}}\ and\ \bibinfo {author} {\bibfnamefont {B.}~\bibnamefont {Xie}},\ }\href {https://arxiv.org/abs/2601.01858} {\bibinfo {title} {A survey of bargmann invariants: Geometric foundations and applications}} (\bibinfo {year} {2026}),\ \Eprint {https://arxiv.org/abs/2601.01858} {arXiv:2601.01858 [quant-ph]} \BibitemShut {NoStop}%
\bibitem [{\citenamefont {Zhang}\ \emph {et~al.}(2025{\natexlab{b}})\citenamefont {Zhang}, \citenamefont {Xie},\ and\ \citenamefont {Tao}}]{BIZL202501}%
  \BibitemOpen
  \bibfield  {author} {\bibinfo {author} {\bibfnamefont {L.}~\bibnamefont {Zhang}}, \bibinfo {author} {\bibfnamefont {B.}~\bibnamefont {Xie}},\ and\ \bibinfo {author} {\bibfnamefont {Y.}~\bibnamefont {Tao}},\ }\bibfield  {title} {\bibinfo {title} {Bargmann-invariant framework for local unitary equivalence and entanglement},\ }\href {https://doi.org/10.1103/s3mp-3kn6} {\bibfield  {journal} {\bibinfo  {journal} {Phys. Rev. A}\ }\textbf {\bibinfo {volume} {112}},\ \bibinfo {pages} {052426} (\bibinfo {year} {2025}{\natexlab{b}})}\BibitemShut {NoStop}%
\bibitem [{\citenamefont {Ma}\ and\ \citenamefont {Shi}(2026)}]{LUBI202607}%
  \BibitemOpen
  \bibfield  {author} {\bibinfo {author} {\bibfnamefont {M.}~\bibnamefont {Ma}}\ and\ \bibinfo {author} {\bibfnamefont {R.}~\bibnamefont {Shi}},\ }\href {https://arxiv.org/abs/2607.16878} {\bibinfo {title} {Bargmann invariants and local unitary equivalence}} (\bibinfo {year} {2026}),\ \Eprint {https://arxiv.org/abs/2607.16878} {arXiv:2607.16878 [quant-ph]} \BibitemShut {NoStop}%
\bibitem [{\citenamefont {Liu}\ \emph {et~al.}(2026)\citenamefont {Liu}, \citenamefont {Chu}, \citenamefont {Zhang}, \citenamefont {Xie},\ and\ \citenamefont {Cui}}]{LPLU25}%
  \BibitemOpen
  \bibfield  {author} {\bibinfo {author} {\bibfnamefont {M.}~\bibnamefont {Liu}}, \bibinfo {author} {\bibfnamefont {Y.}~\bibnamefont {Chu}}, \bibinfo {author} {\bibfnamefont {K.}~\bibnamefont {Zhang}}, \bibinfo {author} {\bibfnamefont {Y.}~\bibnamefont {Xie}},\ and\ \bibinfo {author} {\bibfnamefont {C.}~\bibnamefont {Cui}},\ }\bibfield  {title} {\bibinfo {title} {On local unitary equivalence and local permutation equivalence of bipartite quantum states under permutation transformations},\ }\href {https://doi.org/https://doi.org/10.1016/j.physleta.2025.131223} {\bibfield  {journal} {\bibinfo  {journal} {Physics Letters A}\ }\textbf {\bibinfo {volume} {568}},\ \bibinfo {pages} {131223} (\bibinfo {year} {2026})}\BibitemShut {NoStop}%
\bibitem [{\citenamefont {Shen}\ \emph {et~al.}(2020)\citenamefont {Shen}, \citenamefont {Chen},\ and\ \citenamefont {jun Zhao}}]{inertiays2020}%
  \BibitemOpen
  \bibfield  {author} {\bibinfo {author} {\bibfnamefont {Y.}~\bibnamefont {Shen}}, \bibinfo {author} {\bibfnamefont {L.}~\bibnamefont {Chen}},\ and\ \bibinfo {author} {\bibfnamefont {L.}~\bibnamefont {jun Zhao}},\ }\bibfield  {title} {\bibinfo {title} {Inertias of entanglement witnesses},\ }\href {http://iopscience.iop.org/10.1088/1751-8121/abbec1} {\bibfield  {journal} {\bibinfo  {journal} {Journal of Physics A: Mathematical and Theoretical}\ } (\bibinfo {year} {2020})}\BibitemShut {NoStop}%
\bibitem [{\citenamefont {Shen}\ \emph {et~al.}(2025)\citenamefont {Shen}, \citenamefont {Chen},\ and\ \citenamefont {Bian}}]{rew2024}%
  \BibitemOpen
  \bibfield  {author} {\bibinfo {author} {\bibfnamefont {Y.}~\bibnamefont {Shen}}, \bibinfo {author} {\bibfnamefont {L.}~\bibnamefont {Chen}},\ and\ \bibinfo {author} {\bibfnamefont {Z.}~\bibnamefont {Bian}},\ }\bibfield  {title} {\bibinfo {title} {Detection power of real entanglement witnesses under local unitary equivalence},\ }\href {https://doi.org/10.1103/py8r-jklj} {\bibfield  {journal} {\bibinfo  {journal} {Phys. Rev. A}\ }\textbf {\bibinfo {volume} {112}},\ \bibinfo {pages} {022415} (\bibinfo {year} {2025})}\BibitemShut {NoStop}%
\end{thebibliography}%
